\documentclass[envcountsame]{llncs}

\usepackage{amsmath}
\usepackage{amssymb}
\usepackage{amsfonts}
\usepackage{mathrsfs}
\usepackage{enumerate}
\usepackage{latexsym}
\usepackage{url}
\usepackage{times}
\usepackage{xspace}

\usepackage{verbatim} 

\usepackage{hyperref}

\hypersetup{
colorlinks=false,
pdfauthor={Alexander Kartzow and Philipp Schlicht},
pdftitle={Structures without Scattered-Automatic Presentation},
pdfkeywords={}
}

\newcommand{\domain}[0]{\mathsf{dom}}
\newcommand{\supp}[0]{\mathsf{supp}}
\newcommand{\Cuts}[0]{\mathsf{Cuts}}

\newcommand{\Ac}[0]{\mathcal{A}}
\newcommand{\Bc}[0]{\mathcal{B}}
\newcommand{\Cc}[0]{\mathcal{C}}
\newcommand{\Dc}[0]{\mathcal{D}}
\newcommand{\Fc}[0]{\mathcal{F}}

\newcommand{\Af}{\mathfrak{A}}
\newcommand{\Bf}{\mathfrak{B}}
\newcommand{\Cf}{\mathfrak{C}}
\newcommand{\Ef}{\mathfrak{E}}
\newcommand{\Ff}{\mathfrak{F}}
\newcommand{\Lf}{\mathfrak{L}}
\newcommand{\Kf}{\mathfrak{K}}
\newcommand{\fr}{\mathsf{fr}}
\newcommand{\Gf}{\mathfrak{G}}
\newcommand{\infrank}{\ensuremath{\infty}\text{-}\ensuremath{\mathsf{rank}}}
\newcommand{\rank}{\ensuremath{\mathsf{rank}}}
\newcommand{\Pf}{\mathfrak{P}}
\newcommand{\Pc}{\mathcal{P}}
\newcommand{\Tf}{\mathfrak{T}}
\newcommand{\Tc}{\mathcal{T}}
\newcommand{\N}{\mathbb{N}}
\newcommand{\Z}{\mathbb{Z}}
\newcommand{\Q}{\mathbb{Q}}
\newcommand{\KB}{\mathsf{KB}}
\newcommand{\FC}{\mathsf{FC}}
\newcommand{\FCstar}{\mathsf{FC_*}}
\newcommand{\image}{\mathsf{im}}
\newcommand{\Pgood}{\ensuremath{P_{\text{good}}}}
\newcommand{\nil}{\mathsf{nil}}

\newcommand{\suffix}[1]{\ensuremath{\overset{\rightarrow}{\sqsupseteq}}_{#1}}

\newcommand{\wordsim}[2]{\sim_{{#1}}^{{#2}}}
\newcommand{\run}[2]{\overset{{#1}}{\underset{{#2}}{\longrightarrow}}}
\newcommand{\LEFT}{\ensuremath{{\bf L}}}
\newcommand{\RIGHT}{\ensuremath{{\bf R}}}

\newcommand{\shuffle}{\mathsf{sh}}

\newcommand{\Words}[1]{\ensuremath{W({#1})}}

\begin{document}

\title{Structures without Scattered-Automatic Presentation}

\author{Alexander Kartzow\inst{1} \and Philipp Schlicht\inst{2}
\thanks{The first  author is supported by the DFG  research project GELO.}
\institute{
Institut f\"ur Informatik, Universit\"at Leipzig, Germany
\and
 Mathematisches Institut, Universit\"at Bonn, Germany
\\
\email{kartzow@informatik.uni-leipzig.de, schlicht@math.uni-bonn.de}}}

\maketitle

\begin{abstract}
  Bruy\`ere and Carton lifted the notion of finite automata reading infinite
  words to finite automata reading words with shape an arbitrary
  linear order $\mathfrak{L}$. Automata on finite words can be used to
  represent infinite structures, the so-called 
  word-automatic structures. Analogously, for a linear order
  $\mathfrak{L}$ there is the class of $\mathfrak{L}$-automatic
  structures. In this paper we prove the following limitations on the
  class of $\mathfrak{L}$-automatic structures for a fixed
  $\mathfrak{L}$ of finite condensation rank $1+\alpha$.

  Firstly, no scattered linear order with finite condensation rank
  above  $\omega^{\alpha+1}$ is $\mathfrak{L}$-automatic. In particular,
  every $\mathfrak{L}$-automatic ordinal is  below
  $\omega^{\omega^{\alpha+1}}$. 
  Secondly, we provide bounds on the (ordinal) height of well-founded
  order trees that are 
  $\mathfrak{L}$-automatic. If $\alpha$ is finite or $\mathfrak{L}$ is
  an ordinal, the height of such a tree is bounded by $\omega^{\alpha+1}$. 
  Finally, we separate the class of tree-automatic structures from
  that of $\mathfrak{L}$-automatic structures for any ordinal
  $\mathfrak{L}$: the countable atomless boolean algebra is known to
  be tree-automatic, but we show that it is not
  $\mathfrak{L}$-automatic. 
\end{abstract}

\section{Introduction}
\label{sec:Intro}

Finite automata  play a crucial
role in many areas of computer science. In particular, finite automata
have been used to represent certain classes of possibly infinite
structures. The basic notion of this branch of research is the class of
automatic structures (cf. \cite{KhoussainovN94}): a
structure is automatic if its domain as well 
as its relations are recognised by (synchronous multi-tape) finite
automata processing finite words. This class has the remarkable
property that the first-order theory of any automatic structure is
decidable. One goal in the theory of automatic structures is a
classification of those structures that are automatic
(cf.~\cite{Delhomme04,KhoussainovRS05,DBLP:journals/lmcs/KhoussainovNRS07,DBLP:journals/apal/KhoussainovM09,KuLiLo11}). 
Besides finite automata reading \emph{finite} or \emph{infinite words}
there are also finite automata reading finite or infinite
\emph{trees}. Using such automata as representation of 
structures leads to the notion of tree-automatic structures
\cite{BLumensath1999}. The
classification of tree-automatic structures is less advanced but some
results have been obtained in the last years 
(cf.~\cite{Delhomme04,Huschenbett13,KaLiLo12}).  
Bruy\`ere and Carton \cite{BruyereC01} adapted the notion of finite
automata such that they can process words that have the shape 
of some fixed linear order. If the linear order is countable and
scattered, the corresponding class of languages possesses the 
good closure properties of the class of languages of finite automata
for finite words (i.e., closure under intersection, union, complement, and
projection) and emptiness of a given language is decidable. Thus, these
automata are also well-suited for representing structures. 
Given a fixed 
scattered linear order $\Lf$ this leads
to the notion of $\Lf$-automatic structures. In case that $\Lf$ is an
ordinal Schlicht and Stephan \cite{SchlichtS11} as well as Finkel and
Todorcevic \cite{FinkelT12} studied  the  classes of 
$\Lf$-automatic ordinals and $\Lf$-automatic linear orders.
Here we  study  $\Lf$-automatic linear orders for any 
scattered linear order $\Lf$ and we study 
$\Lf$-automatic
well-founded order forests,i.e., forests (seen as partial orders)
without infinite branches.
\begin{enumerate}
\item If a linear order is $\Lf$-automatic and $\Lf$ has finite
  condensation rank at most $1+\alpha$, then it is a finite sum of  linear
  orders of condensation rank below $\omega^{\alpha+1}$. As already
  shown in \cite{SchlichtS11}, this bound is optimal.
\item If a well-founded  order forest is $\Lf$-automatic for 
  some ordinal $\Lf$, then its ordinal height is bounded by $\Lf\cdot
  \omega$.
  
  If a well-founded  order forest is $\Lf$-automatic for $\Lf$ 
  some linear order of condensation rank $n\in\N$, 
  then its ordinal height is bounded by $\omega^{n+1}$. 

  These two bounds are optimal.
\item A well-founded $\Lf$-automatic order forest has ordinal height
  bounded by $\omega^{\omega\cdot (\alpha+1)}$ where $\alpha$ is the
    finite condensation rank of $\Lf$. 
\end{enumerate}
In order to prove  Claims 1 and 3 we observe that
the notion of \emph{finite-type products} from \cite{SchlichtS11} and
the notion of \emph{sum-augmentations of tamely colourable box-augmentations} 
from \cite{KaLiLo12,Huschenbett13}, even though defined in completely different
terms, have a common underlying idea. We introduce a new notion of
tamely colourable sum-of-box augmentations that refines both notions
and allows to prove a variant of Delhomm\'{e}'s decomposition method
(cf.~\cite{Delhomme04}) for the case of $\Lf$-automatic structures. 
The main results then follow as corollaries using results from
\cite{Huschenbett13} and \cite{KaLiLo12}. For the other two results,
we provide an $\Lf$-automatic scattered linear ordering of all
$\Lf$-shaped words if $\Lf$ has finite condensation rank
$n\in\N$ or if $\Lf$ is an ordinal. Extending work from \cite{KuLiLo11}, 
we provide a connection between the height of a
tree and the finite condensation rank of its Kleene-Brouwer ordering
(with respect to this $\Lf$-automatic ordering) that allows to derive
the better bounds  stated in Claim 2.

As a very sketchy summary of these results, one could say that we
adapt techniques 
previously used on trees to use them on linear orders. This raises the
question whether there is a deeper connection between $\Lf$-automatic
structures and tree-automatic structures. It is known that
all $\omega^n$-automatic structures are tree-automatic
(cf. \cite{FinkelT12}). Moreover, from
\cite{SchlichtS11} and \cite{Delhomme04} it follows that
$\omega^{\omega^\omega}$ is $\omega^\omega$-automatic but not
tree-automatic. It is open so far whether every tree-automatic
structure is $\Lf$-automatic for some linear order $\Lf$. We make a
first step towards a negative answer by showing that the countable atomless
boolean algebra is not $\Lf$-automatic for any ordinal $\Lf$ (while it
is tree-automatic \cite{BGR11}).

\section{Preliminaries}
\label{sec:preliminaries}

\subsection{Scattered Linear Orders}
\label{sec:ScatteredOrders}

In this section, we recall basic notions concerning scattered linear
orders. For a detailed introduction, we refer the reader to
\cite{rosenstein82}.
A linear order $(L, \leq)$ is \emph{scattered} if there is no
embedding of the rational numbers  into $(L, \leq)$. 

Given a scattered linear order $\Lf = (L, \leq)$, an equivalence relation
$\sim$ is called a \emph{condensation} if each $\sim$ class is an
interval of 
$\Lf$. We then write $\Lf/{\sim}:=(L/{\sim}, \leq')$ for
the linear order of the 
$\sim$ classes induced by $\leq$ (i.e., for $\sim$-classes $x,y$, 
$ x\leq' y$ iff there are $k\in x, l\in y$ such that $k
 \leq l$). 
As usual, for $\Lf$ a scattered linear order and $l,l'$ elements of
$\Lf$, we 
write $[l,l']$ for the closed interval between $l$ and $l'$. 
For each ordinal $\alpha$
we define the \emph{$\alpha$-th condensation} $\sim_\alpha$ by
$x \sim_0 y$ iff $x=y$,   $x \sim_{\alpha+1} y$ if the closed interval
$[x,y]$ in $\Lf/{\sim_\alpha}$ is finite and for a limit ordinal $\beta$, 
$x \sim_\beta y$ if there is an $\alpha<\beta$ such that
$x\sim_\alpha y$. 
The \emph{finite condensation rank}
$\FC(\Lf)$ is the minimal ordinal $\alpha$ such that 
$\Lf/{\sim_\alpha}$ is a one-element order. 
We also let $\FCstar(\Lf)$ be the minimal ordinal $\alpha$ such that
$\Lf/{\sim_\alpha}$ is a finite order. There is such an ordinal $\alpha$ if and only if $\Lf$ is scattered. 
It is obvious from these definitions that 
$\FCstar(\Lf)\leq \FC(\Lf) \leq \FCstar(\Lf)+1$.

As usual, for  a linear order $\Lf=(L, \leq)$ and 
a sequence of linear orders $(\Lf_i)_{i\in\Lf}$ we denote by
$\sum_{i\in L} \Lf_i$ the 
$\Lf$-sum of the $(\Lf_i)_{i\in\Lf}$. 

We conclude this section by recalling the notion of Dedekind cuts of a
linear order.
Let $\Lf=(L, \leq)$ be a linear order. A \emph{cut} of $\Lf$ is a
pair $c=(C, D)$ where $C$ is a downward closed subset $C\subseteq L$
and $D=L\setminus C$. We write $\Cuts(\Lf)$ for the 
\emph{set of all cuts  of $\Lf$}. 
For cuts $c,d$, we say that $c$ and $d$ are the consecutive cuts around
some $l\in L$ if $c=(C, D)$ and $d=(C', D')$ such that 
$C  =\{x\in L\mid x < l\}$ and
$C' =\{x\in L\mid x \leq l\}$. 
$\Cuts(\Lf)$ can be naturally equipped with an order (also denoted by
$\leq$) via
$c=(C,D) \leq d=(C',D')$ if $C\subseteq C'$. 
We say a cut $c=(C,D)$ has no
direct predecessor (or direct successor), if it has no direct predecessor
(or direct successor, respectively) with respect  to $\leq$. 
Let us finally introduce a notation for values appearing arbitrarily
close to some cut (from below or from above, respectively). 
\begin{definition}
  Let $\Lf=(L, \leq)$ be a linear order, and $w:\Cuts(\Lf)\to A$. For
  $c=(C,D)\in\Cuts(\Lf)$, set 
  $\lim_{c^{-}} w:= \{a\in A \mid \forall l\in C
    \exists l'\in C\quad l\leq l'\text{ and }w(l')=a\}$ \text{ and}
    $\lim_{c^{+}} w:= \{a\in A \mid \forall l\in D
    \exists l' \in D\quad l'\leq l\text{ and } w(l')= a\}$.
\end{definition}

\subsection{Automata for Scattered Words and 
  Scattered-Automatic Structures} 
\label{sec:Automata}

For this section, we fix an arbitrary linear order
$\Lf=(L, \leq)$.

\begin{definition}
  Let  $\Sigma_\diamond$
  be some finite alphabet with $\diamond\in\Sigma_\diamond$. 
  An $\Lf$-word (over $\Sigma$) is a map $L\to \Sigma_\diamond$. 
  An $\Lf$-word $w$ is \emph{finite} if 
  the support $\supp(w):=\{l\in L\mid w(l)\neq\diamond\}$
  of $w$ is finite.
  $\Words{\Lf}$ denotes  the set of $\Lf$-words.
\end{definition}

The usual notion of a convolution of finite words used in automata
theory can be easily lifted to the case of $\Lf$-words.

\begin{definition}
  Let $w_1,w_2$ be $\Lf$-words over alphabets $\Sigma_1$ and
  $\Sigma_2$, respectively. 
  The \emph{convolution}
  $w_1\otimes w_2$ is
  the $\Lf$-word over  $\Sigma_1\times \Sigma_2$ given by
  $[w\otimes v](l) := (w(l),v(l))$. 
\end{definition}

We  recall Bruy\`ere and Carton's definition
of automata for $\Lf$-words \cite{BruyereC01}. Then we introduce the
notion of (finite word) $\Lf$-automatic structures generalising the
notion of ordinal-automatic structures from \cite{SchlichtS11}.

\begin{definition}
  An \emph{$\Lf$-automaton} is a tuple
  $\Ac= (Q, \Sigma, I, F, \Delta)$ where $Q$ is a finite set of states,
  $\Sigma$ a finite alphabet, $I \subseteq Q$ the initial and $F\subseteq Q$
  the final states and $\Delta$ is a subset of
    $
    \left(Q\times \Sigma \times Q \right) \cup
    \left( 2^Q\times Q  \right) \cup   \left( Q\times 2^Q  \right)
    $
  called the transition relation. 
\end{definition}
Transitions in $Q\times \Sigma \times Q$ are called 
\emph{successor transitions},
transitions in $2^Q\times Q$ are 
called \emph{right limit transitions}, and
transitions in $Q\times 2^Q$ are 
called \emph{left limit transitions}.

\begin{definition}
  A \emph{run} of $\Ac$ on the $\Lf$-word $w$ is a map 
  $r:\Cuts(\Lf) \to Q$ such that
  \begin{itemize}
  \item $ \left( r(c), w(l), r(d)\right)\in \Delta$ for
    all $l\in L$ and all consecutive cuts $c,d$ around $l$,
  \item $(\lim_{c^-} r, r(c))\in\Delta$
    for all cuts $c\in\Cuts(\Lf)\setminus \{(\emptyset,L)\}$ without
    direct  predecessor, 
  \item $(r(c), \lim_{c^+} r )\in\Delta$
    for all cuts $c\in\Cuts(\Lf)\setminus \{(L,\emptyset)\}$ without
    direct successor. 
  \end{itemize}
  The run $r$ is \emph{accepting} if 
  $r((\emptyset,L))\in I$ and
  $r((L,\emptyset))\in F$. 
  The \emph{language} of $\Ac$ consists of all $\Lf$-words $w$ such
  that there is an accepting run of $\Ac$ on $w$. 
  For some $\Lf$-word $w$ and states $q,q'$ of $\Ac$ we write
  $q \run{w}{\Ac} q'$ if there is a run $r$ of $\Ac$ on $w$ such that
  $r((\emptyset,L))=q$ and $r((L,\emptyset))=q'$. 
\end{definition}

\begin{example} \label{exa:FiniteWordsRec} 
  The following $\Lf$-automaton accepts the set of finite
  $\Lf$-words over the alphabet $\Sigma$. Let $\Ac = (Q, \Sigma, I, F,
  \Delta)$
  with $Q= \{e_l, e_r, n, p\}$, $I=\{n\}, F=\{n,p\}$, and
  \begin{align*}
    \Delta&= \{ (n, \diamond, n), (p, \diamond, n)\}\cup 
    \{ (n,\sigma, p), (p, \sigma,p)\mid
    \sigma\in\Sigma\setminus\{\diamond\} \} \\
    &\cup
    \{(\{n\}, n), (n, \{n\}), 
    (p, \{n\}), 
    (\{p\}, e_l), (e_r, \{p\}),   
    (\{n,p\}, e_l), (e_r, \{n,p\})\}.
  \end{align*}
  For each $w\in \Words{\Lf}$, 
  \mbox{  $r((C,D))=
    \begin{cases}
      p & \text{if } \max(C) \text{ exists and } \max(C)\in\supp(w)\\
      n & \text{otherwise},
    \end{cases}$}
  defines an accepting run if $w$ is a finite $\Lf$-word.
  On an $\Lf$-word $w$ with infinite support, the successor
  transitions require infinitely many occurrences of state $p$. But
  then some limit position is marked with an error state $e_l$ or
  $e_r$ (where $l$ means 'from left' and $r$ 'from right') and the run
  cannot be continued (see Appendix \ref{sec:FiniteWordsRecognisable}
  for details). 
\end{example}

Automata on words (or infinite words or trees or infinite trees) have
been applied fruitfully for representing structures. This can be
lifted to the setting of $\Lf$-words and leads to the notion of
(oracle)-$\Lf$-automatic structures.   

\begin{definition} \label{definition: L-automatic structure} 
  Fix an $\Lf$-word $o$ (called an oracle).
  A structure $\Af=(A,R_1,R_2, \dots, R_m)$ is
  \emph{$\Lf$-$o$-automatic} 
  if there are $\Lf$-automata
  $\Ac, \Ac_1, \dots, \Ac_m$ such that
  \begin{itemize}
  \item $\Ac$ represents the domain of $\Af$ in the sense that
    $A= \{w \mid w\otimes o \in L(\Ac)\}$, and
  \item for each $i\leq m$, $\Ac_i$ represents $R_i$ in the sense that
    $R_i= \{ (w_1,w_2\dots, w_{r_i})\mid 
    w_1\otimes w_2\otimes\dots\otimes w_{r_i}\otimes o \in L(\Ac_i)\}$,
    where $r_i$ is the arity of relation $R_i$. 
  \end{itemize}
  We say that an $\Lf$-$o$-automatic structure is \emph{finite word}
  $\Lf$-$o$-automatic if its domain consists only of finite $\Lf$-words.
  Let $\Fc_\Lf$ denote the class of all
  finite word $\Lf$-oracle-automatic graphs. 
\end{definition}
For the constantly $\diamond$-valued oracle $o$ ($\forall x\in\Lf\
o(x)=\diamond$),
we call an $\Lf$-$o$-automatic structure  $\Lf$-automatic. 
We call some structure $\Af$  scattered-automatic
(scattered-oracle-automatic, respectively) if there is some 
scattered linear order $\Lf'$ (and some oracle $o$) such that $\Af$ is
finite word $\Lf'$-automatic ($\Lf'$-$o$-automatic, respectively).

Rispal and Carton \cite{RispalC05} showed that $\Lf$-oracle-automata
are closed under complementation if $\Lf$ is countable and
scattered which implies the following Proposition.
\begin{proposition}
  If $\Lf$ is a countable scattered linear order, 
  the set of finite word $\Lf$-$o$-automatic structures is closed under
  first-order definable relations. 
\end{proposition}

\subsection{Order Forests}

\begin{definition}
  An (order) \emph{forest} is a partial order $\Af=(A, \leq)$ such that
  for each $a\in A$, the set $\{a'\in A\mid a\leq a'\}$ is a finite
  linear order. 
\end{definition}

Later we study the rank (also called ordinal height) of
$\Lf$-automatic well-founded
forests. For this purpose we recall the definition of rank.
Let $\Af=(A, \leq)$ be a well-founded partial order.
Setting
$\sup(\emptyset) = 0$ we define the \emph{rank} of $\Af$ by
$\rank(a, \Af)=\sup\{\rank(a',\Af)+1\mid a'<a \in A\}$ and
$\rank(\Af)=\sup\{\rank(a,\Af)+1 \mid a\in A\}$.

\section{Sum- and Box-Augmentation Technique} 
\label{sec:BoxDecomposition}

Delhomm\'{e} \cite{Delhomme04} characterised the set of ordinals that
can be represented by finite tree-automata. His results relies on a
decomposition of definable substructures into \emph{sum-} and
\emph{box-augmentations}. Huschenbett \cite{Huschenbett13} and  
Kartzow et al.~\cite{KaLiLo12} introduced a refined
notion of \emph{tamely colourable} 
box-augmentations in order to bound the ranks of tree-automatic
linear orders and well-founded order trees, respectively. We first
recall the definitions and then show that the decomposition technique
also applies to finite word scattered-oracle-automatic structures.

Before we go into details, let us sketch the ideas
underlying the sum- and box-augmentation technique. 
Given an $\Lf$-$o$-automatic structure $\Af$ with domain $A$ and
some automaton $\Ac$ (called \emph{parameter automaton}) that
recognises a subset of $A\times \Words{\Lf}$, let us denote by
$\Af_p$ the substructure of $\Af$ induced by $\Ac$ and $p$, i.e., with
domain $\{a\in A\mid a\otimes p\in L(\Ac)\}$. The main proposition of
this section says that there is a certain class $\Cc$ of structures
(independent of $p$) 
such that each $\Af_p$ is a tamely colourable sum-of-box augmentation
of structures 
from $\Cc$. $\Cc$ consists of finitely many $\Lf$-oracle-automatic
structures and scattered-oracle-automatic structures where the
underlying scattered linear order has finite condensation rank
strictly below that of $\Lf$. 
This allows to compute bounds on structural parameters (like finite
condensation rank of linear orders or ordinal height of well-founded
partial orders) by induction on the rank of $\Lf$. We say a structural
parameter $\varphi$ is compatible with sum-of-box augmentations if for
$\Af$ a sum-of-box augmentation of $\Af_1,
\dots, \Af_n$, there is a bound on $\varphi(\Af)$ in terms of
$\varphi(\Af_1), \dots, \varphi(\Af_n)$. 
The decomposition result tells us that some $\Lf$-automatic
structure $\Af$ is (mainly) a sum of boxes  of scattered-automatic
structures where the underlying orders have lower ranks. Thus, by
induction hypothesis $\varphi$ is bounded on these building blocks of
$\Af$. Thus, $\varphi(\Af)$ is also bounded if $\varphi$ is compatible
with sum- and box-augmentations.

\subsection{Sums and Boxes}
\label{sec:sumBox}

The next definition recalls the notion of sum- and
box-augmentations. We restrict the presentation to structures with
one binary relation (but the general case is analogous).

\begin{definition}
  \begin{itemize}
  \item A structure $\Af$ is a \emph{sum-augmentation} of structures
    $\Af_1, \dots, 
    \Af_n$ if the domain of $\Af$ can be partitioned into $n$ pairwise
    disjoint sets such
    that the substructure induced by the $i$-th set is isomorphic to
    $\Af_i$. 
  \item 
    A structure $\Af=(A, \leq^A)$ is a \emph{box-augmentation} of structures
    $\Bf_1=(B_1, \leq^{B_1}),$  $\dots$,\mbox{$\Bf_n=(B_n, \leq^{B_n})$} if
    there is  a bijection 
    $\eta: \prod_{i=1}^n B_i \to A$ such that 
    for all $1\leq j \leq n$ and all $\bar b= (b_1, \dots, b_n)\in
    B_1\times\dots\times B_n$ 
    \begin{align*}
      \Bf_j \simeq \Af{\restriction}_{\eta(\{b_1\}\times \dots\times
        \{b_{j-1}\} \times B_j \times \{b_{j+1}\} \times \dots \times \{b_n\})}.
    \end{align*}
  \item   
    Let $\Cc_1, \dots, \Cc_n$ be classes of structures. A structure
    $\Af$ is a \emph{sum-of-box augmentation} of
    $(\Cc_1, \dots, \Cc_n)$ if $\Af$ is a 
    sum-augmentation of 
    structures $\Bf_1, \dots, \Bf_k$ such
    that each $\Bf_j$ is a box-augmentation of  structures $\Cf_{j,1},
    \dots, \Cf_{j,n}$ with  $\Cf_{j,i}\in \Cc_i$. 
  \end{itemize}
\end{definition}

\begin{definition}
  Let $\Af=(A, \leq)$ be a sum-of-box augmentation of structures 
  $\Bf_{i,j}=(B_{i,j}, \leq_{i,j} )$ via the map 
$
    \eta: \bigsqcup_{i=1}^n  \prod_{j=1}^k  B_{i,j}  \to A.
$
  This sum-of-box augmentation is called 
  \emph{tamely colourable}
  if for each $1\leq j \leq k$ there is a function 
  $\varphi_j: (\bigsqcup_{i=1}^n B_{i,j})^2 \to C_j$
  with a finite range $C_j$  such that 
  the $(\varphi_j)_{1\leq j \leq k}$ determine the edges of $\Af$ in
  the sense that there is 
  a set $M \subseteq \prod_{j=1}^k C_j$ such that
  $\eta(b_1,\dots, b_k) \leq \eta(b_1',\dots,b_k')$ iff
  $\left( \varphi_1(b_1,b_1'), \dots, \varphi_k(b_k,b_k')\right)\in M$. 
\end{definition}

\subsection{Decomposition of Scattered-Automatic-Structures}
\label{sec:decomp-scattered}

In this section, we prove that the sum- and box-augmentation
technique applies to finite word scattered-oracle-automatic
structures. 
Fix an arbitrary scattered order $\Lf$ with 
$\FC(\Lf)=\alpha \geq 1$. Assume that 
$\Lf = \sum_{z\in \Z} \Lf_z$ where each $\Lf_z$ is a
(possibly empty) suborder with $\FC(\Lf_z) < \alpha$.
We first introduce notation concerning definable subgraphs.

\begin{definition}
  Let $o\in\Words{\Lf}$ be some oracle.
  Let $\Gf=(V, E)$ be a finite word $\Lf$-$o$-automatic graph. 
  For each parameter automaton $\Ac$ and parameter $p\in \Words{\Lf}$,
  we write  $\Gf^\Ac_p$ for the induced subgraph of $\Gf$ with domain   
$
    V^\Ac_p:=\{w\in V\mid w\otimes p \in L(\Ac)\}.     
$
\end{definition}
We write $\Gf_p$ and $V_p$  for $\Gf^\Ac_p$ and $V^\Ac_p$  if $\Ac$ is
clear from the context.

\begin{definition}
  Let $c_0=(C_0,D_0)$ and $c_1=(C_1,D_1)$ be cuts of $\Lf$. 
  For  a finite  $\Lf$-word $w$ we say $w$ is a 
  $(c_0,c_1)$-parameter if $\supp(w)\subseteq D_0\cap C_1$, i.e., 
  the support of $w$ is completely between $c_0$ and
  $c_1$. 
\end{definition}

\emph{For the rest of this section, we fix two numbers 
$z_0<z_1\in \Z$ and define the  cuts
$c_0:=(\sum_{z < z_0} \Lf_z, \sum_{z\geq z_0} \Lf_z)$ and
$c_1:=(\sum_{z \leq  z_1} \Lf_z, \sum_{z > z_1} \Lf_z)$.
We also define the scattered orders 
$\Lf_{\LEFT}:=\sum_{z < z_0} \Lf_z$ and
$\Lf_{\RIGHT}:=\sum_{z > z_1} \Lf_z$.}
The main result of this section is a uniform sum-of-box decomposition 
of all substructures defined by a given parameter automaton. 

\begin{theorem} \label{thm:Decomposition}
  Let  $\Gf$ be 
  some finite word $\Lf$-oracle-automatic graph $(V,E)$
  where $E$ is recognised by some automaton $\Ac_E$ with state set
  $Q_E$ and 
  let $\Ac$ be a parameter automaton with state set $Q$. 
  There are
  \begin{itemize}
  \item a set $\Cc_\LEFT$ of $\exp(\lvert Q \rvert^2+2\lvert Q_E\rvert^2)$
    many $\Lf_{\LEFT}$-oracle-automatic graphs, and
  \item  a set $\Cc_\RIGHT$ of $\exp(\lvert Q \rvert^2+2\lvert Q_E\rvert^2)$
    many $\Lf_{\RIGHT}$-oracle-automatic graphs,
  \end{itemize}
  such that for each $(c_0,c_1)$-parameter $p$ the subgraph 
  $\Gf^{\Ac}_p$ is a tamely-colourable sum-of-box-augmentation
  of $(\Cc_\LEFT, \Fc_{\Lf_{z_0}}, \Fc_{\Lf_{z_0+1}}, \dots, 
  \Fc_{\Lf_{z_1}}, \Cc_\RIGHT)$. \footnote{Recall that $\Fc_{\Lf}$ is the class of all finite word $\Lf$-oracle-automatic graphs, see Definition \ref{definition: L-automatic structure}.} 
\end{theorem} 
\begin{proof} 
  Let $o$ be the oracle such that $\Gf$ is finite word 
  $\Lf$-$o$-automatic. 
  By definition, we can write $\Lf$ as the sum 
$
    \Lf_{\LEFT} + \Lf_{z_0} + \Lf_{z_0+1} +\dots +
    \Lf_{z_1} + \Lf_{\RIGHT}.   
$
  Induced by this decomposition there is a decomposition of 
  any $\Lf$-word $w$ as 
$    w = w_{\LEFT} w_{z_0} w_{z_0+1} \dots w_{z_1} w_{\RIGHT}
$
  such that $w_j$ is an $\Lf_j$-word. In particular, our parameter
  and oracle decompose as
  \begin{align*}
    p = p_{\LEFT} p_{z_0} p_{z_0+1} \dots p_{z_1} p_{\RIGHT}\text{\quad
      and\quad }
    o = o_{\LEFT} o_{z_0} o_{z_0+1} \dots o_{z_1} o_{\RIGHT}.
  \end{align*}
  Independently of the choice of the $(c_0,c_1)$-parameter 
  $p$, $p_{\LEFT}$ and
  $p_{\RIGHT}$ are constant functions (with value $\diamond$).
  
  In order to construct a sum-of-box decomposition of $\Gf_p$, we
  first define the building blocks of this decomposition. 
  For this purpose, we define equivalence relations 
  $\wordsim{p\otimes o}{i}$ for each 
  $i\in\{\LEFT, \RIGHT, z_0, {z_0+1}, \dots, z_1\}$ on $\Lf_i$-words
  as follows.   
  For $\Lf_i$-words $w,w'$ set  $w \wordsim{p\otimes o}{i} w'$ if and only if
  \begin{enumerate}
  \item for all $q,q'\in Q$ \quad
    $q \run{w\otimes p_i\otimes o_i}{\Ac} q'  \Longleftrightarrow 
     q \run{w'\otimes p_i \otimes o_i}{\Ac} q'$ and
   \item for all $q,q'\in Q_E$ \quad
    $q \run{w\otimes w \otimes o_i}{\Ac_E} q'  \Longleftrightarrow 
     q \run{w'\otimes w' \otimes o_i}{\Ac_E} q'$.
  \end{enumerate}
  Note that for fixed $i,p,o$ there are at most $\exp({\lvert Q\times Q\rvert+
    \lvert Q_E\times Q_E \rvert})$ many $\wordsim{p\otimes o}{i}$
  equivalence classes.   
  As domains of the $\alpha_i$-oracle-automatic building blocks of our
  decomposition we use the sets
  $K(i,w,p,o):=\{ x\mid x \wordsim{p\otimes o}{i} w\}$ for each
  $\Lf_i$-word $w$. 
  We augment this notation by writing $K(i,v,p,o):=K(i,w,p,o)$ for $\Lf$-words $v$, where $w$ is the restriction of $v$ to $\Lf_i$. 
  Now for each $M\subseteq Q_E\times Q_E$ we define a structure
  $\Kf^M(i,w,p,o)=(K(i,w,p,o), E^M)$ where
  $(w_1,w_2)\in E^M$ if 
  $w_1,w_2\in K(i,w,p,o)$ and there is a  $(q,q')\in M$ such that
  $q \run{w_1\otimes w_2 \otimes o}{\Ac_E} q'$. 
  Recall that $p_{\LEFT}$ and $p_{\RIGHT}$ are independent
  of the concrete choice of the $(c_0,c_1)$-parameter $p$ whence (for
  fixed $o$) the sets
  \begin{align*}
    \Cc_{\LEFT}:=\left\{\Kf^M(\LEFT,w,p,o)\mid M \subseteq Q_E\times Q_E, p
     \text{ a }  (c_0,c_1)\text{-parameter}  \right\}\\
    \Cc_{\RIGHT}:=\left\{\Kf^M(\RIGHT,w,p,o)\mid M \subseteq Q_E\times Q_E, p
      \text{ a } (c_0,c_1)\text{-parameter}  \right\}
  \end{align*}
  have each at most $\exp(\lvert Q \rvert^2+2\lvert Q_E\rvert^2)$ many
  elements (up to isomorphisms). 

  Our next goal is the definition of the function $\eta$ that
  witnesses the decomposition claimed in this theorem. 
  For this purpose, let $\wordsim{p\otimes o}{}$ denote the
  equivalence on 
  $\Lf$-words that is the product of the 
  $\wordsim{p\otimes o}{i}$.\footnote{
    Thus,  for $w=w_{\LEFT}w_{z_0}w_{z_0+1}\dots w_{z_1} w_{\RIGHT}$ and
    $v=v_{\LEFT}v_{z_0}v_{z_0+1}\dots v_{z_1} v_{\RIGHT}$ we have
    $w\wordsim{p\otimes o}{} v$ iff
    $w_i\wordsim{p\otimes o}{i} v_i$ for all 
    $i\in\{\LEFT,\RIGHT, z_0,  z_0+1, \dots, z_1\}$.
  }
  Let
  \begin{align*}
    \eta: &\bigsqcup_{[w]\in V_p /{\wordsim{p\otimes o}{}}}
    K(\LEFT, w,p,o) \times \left(\prod_{i=z_0}^{z_1} K(i,w,p,o)\right )
    \times K(\RIGHT, w,p,o)  \longrightarrow V_p\\
    &(x_{\LEFT}, x_{z_0}, x_{z_0+1}, \dots, x_{z_1}, x_{\RIGHT}) \mapsto 
    x:=x_{\LEFT}x_{z_0} x_{z_0+1}\dots x_{z_1}x_{\RIGHT}. 
  \end{align*}
  It follows from the definitions that $\eta$ is a well-defined
  bijection (using the fact that some $\Lf$-word $x$ belongs to
  $V_p$ iff there is a run 
  \begin{align*}
    q_I \run{x_{\LEFT}\otimes p_{\LEFT} \otimes
    o_{\LEFT}}{\Ac} q_{z_0}  \run{x_{z_0}\otimes p_{z_0} \otimes
    o_{z_0}}{\Ac} q_{z_0+1} \dots q_{z_1} \run{x_{\RIGHT} \otimes p_{\RIGHT}
      \otimes o_{\RIGHT}}{\Ac} q_F    
  \end{align*}
  for some initial state $q_I$ and a final state $q_F$). 
  
  In order to finish the proof, we show that  $\Gf_p$ is a
  tamely-colourable sum-of-box-augmentation of 
  $(\Cc_\LEFT, \Fc_{\Lf_{z_0}}, \Fc_{\Lf_{z_0+1}}, \dots, 
  \Fc_{\Lf_{z_1}}, \Cc_\RIGHT)$ via
  $\eta$. For any $w\in V_p$, let
  $\Ff_w$ be the restriction of $\Gf_p$ to 
  $\eta\left(    K(\LEFT, w,p,o) \times \left(\prod_{i=z_0}^{z_1}
      K(i,w,p,o)\right )  \times K(\RIGHT, w,p,o)\right)$. 
  It is clear that $\Gf_p$ is a sum augmentation of 
  $(\Ff_{w_1}, \Ff_{w_2}, \dots, \Ff_{w_k})$ for $w_i$ representatives of
  the $\wordsim{p\otimes o}{}$-classes. 
  From now on 
  let  $I_E (F_E)$ denote the initial (final) states of
  $\Ac_E$.
  \begin{enumerate}
  \item Fix \mbox{$w= w_{\LEFT} w_{z_0} w_{z_0+1} \dots w_{z_1} w_{\RIGHT}\in
      V_p$}. We show that $\Ff_w$ is a box-augmen\-tation of
    $(\Cc_\LEFT, \Fc_{\Lf_{z_0}}, \Fc_{\Lf_{z_0+1}}, \dots, 
    \Fc_{\Lf_{z_1}}, \Cc_\RIGHT)$. For this purpose, fix 
    $i\in \{\LEFT, \RIGHT, z_0, z_0+1, \dots, z_1\}$  and 
    let $\overset{\leftarrow}{w}:=  w_{\LEFT}\dots w_{i-1}$,
    $\overset{\leftarrow}{o} :=  o_{\LEFT}\dots o_{i-1}$,
    $\overset{\rightarrow}{w} :=  w_{i+1}\dots w_{\RIGHT}$, and 
    $\overset{\rightarrow}{o} :=  o_{i+1}\dots o_{\RIGHT}$. 
    Let $M_i$ be the set defined by 
    \begin{equation}
      \label{eq:boxelement}
      (q_1,q_2)\in M_i \Longleftrightarrow \exists q_I\in I_E, q_F\in F_E
      \quad
      q_I \run
      {\overset{\leftarrow}{w} \otimes 
        \overset{\leftarrow}{w} \otimes 
        \overset{\leftarrow}{o}}{ \Ac_E} q_1 \text{ and }
      q_2 \run{
        \overset{\rightarrow}{w} \otimes 
        \overset{\rightarrow}{w} \otimes 
        \overset{\rightarrow}{o}}{ \Ac_E} q_F.
    \end{equation}
    The function
    \begin{align*}
      \eta^{w}_i: K(i,w,p,o) \to  V_p, \quad
      x_i \mapsto w_{\LEFT} w_{z_0} w_{z_0+1} \dots w_{i-1} x_i w_{i+1}
    \dots w_{z_1} w_{\RIGHT}
    \end{align*}
    embeds $\Kf^{M_i}(i,w,p,o)$ into $\Gf_p$
    because
    \begin{align*}
      &\forall x_i,y_i\in K(i,w,p,o)\quad (x_i,y_i)\in E^{M_i} \\
      \Leftrightarrow &\exists (q_1,q_2)\in M_i \quad  q_1 \run{x_i\otimes
        y_i \otimes o_i}{\Ac_E} q_2\\
      \overset{\eqref{eq:boxelement}}{\Leftrightarrow} 
      &\exists q_I\in I_E, q_F\in F_E \quad
      q_I \run{ \overset{\leftarrow}{w} \otimes 
        \overset{\leftarrow}{w} \otimes \overset{\leftarrow}{o}}{ \Ac_E} q_1
      \run{x_i\otimes  y_i \otimes o_i}{\Ac_E} q_2
      \run{\overset{\rightarrow}{w} \otimes 
        \overset{\rightarrow}{w} \otimes
        \overset{\rightarrow}{o}}{ \Ac_E} q_F \\
      \Leftrightarrow 
      &\left(\eta^w_i(x_i), \eta^w_i(y_i)\right) \in E.
    \end{align*}
  \item We show that the decomposition is tamely colourable. 
    For all $j\in \{\LEFT, \RIGHT, z_0, z_0+1, \dots, z_1\}$,
    let
      $c_j: (\bigsqcup_{[w]\in V_p/\wordsim{p\otimes o}{}} K(j,w,p,o))^2 \to
      Q_E^2$ be the colouring function satisfying
      \mbox{$c_j(x_j,y_j) := \{ (q,q')\in A_E \mid 
      q\run{x_j\otimes y_j\otimes o_j}{\Ac_E} q'\}$.}
    The colour functions  $(c_j)_{j\in\{\LEFT, \RIGHT, z_0, z_0+1,
      \dots, z_1)}$ determine  $E$ because for
    $w = w_\LEFT w_{z_0} w_{z_0+1} \dots w_{z_1} w_\RIGHT$ and
    \mbox{$v = v_\LEFT v_{z_0} v_{z_0+1} \dots v_{z_1} v_\RIGHT$, }
    \begin{align*}
    &(w_\LEFT w_{z_0} w_{z_0+1} \dots w_{z_1} w_\RIGHT,
    v_\LEFT v_{z_0} v_{z_0+1} \dots v_{z_1} v_\RIGHT)\in E\\
    \Longleftrightarrow 
    &\exists q_0, \dots, q_k\in Q_E \ 
    \left(
    \begin{aligned}
      &q_0\in I_E, q_k\in F_E,
      \text{ and } \\     
      &q_0 \run{w_\LEFT\otimes v_\LEFT \otimes o_\LEFT}{\Ac_E} q_1 
      \run{w_{z_0}\otimes v_{z_0} \otimes o_{z_0}}{\Ac_E} q_2   
      \dots 
      q_{k-1} 
      \run{w_{\RIGHT}\otimes v_{\RIGHT} \otimes o_{\RIGHT}}{\Ac_E} q_k
    \end{aligned} \right)
    \\
    \Longleftrightarrow
    &\exists q_0, \dots, q_k\in Q_E \ 
    \left(
      \begin{aligned}
        &q_0\in I_E, q_k\in F_E,   \text{ and }\\
        &(q_{i-1},q_{i})\in c_j(w_j,v_j)        
       \text{ with } j=
        \begin{cases}
          \LEFT &\text{if }i=1,\\
          \RIGHT &\text{if } i={k},\\
          z_0+m &\text{if } i=m  
        \end{cases}
      \end{aligned} \right) . 
    \end{align*}
  \end{enumerate} 
\end{proof}

\section{Bounds on 
  Scattered-Oracle-Automatic Structures} 
\label{sec:LinOrd}

\subsection{FC-Ranks of Linear-Orders}
\label{sec:fcrankbound}

In this section, we first study the question which scattered linear orders
are $\Lf$-oracle-automatic for a fixed order $\Lf$. We provide a sharp
bound on the $\FC$-rank. For the upper bound we lift Schlicht and
Stephan's result \cite{SchlichtS11} using our new sum- and box-decomposition
from the case where $\Lf$ is an ordinal (detailed proof in Appendix
\ref{sec:ProofOnScatteredOrders}):  

\begin{theorem} \label{thm:BoundsonScatteredOrders}
  Let $\Lf$ be a scattered order of $\FCstar$ rank
  $1+\alpha$ ($0$, respectively) for some ordinal $\alpha$. 
  Then every finite word $\Lf$-oracle-automatic scattered linear order
  $\Af$ satisfies 
  \mbox{$\FCstar(\Af) < \omega^{\alpha+1}$} ($\FCstar(\Af) <\omega^0=1$, respectively).
\end{theorem}

If $\Lf$ is an ordinal of the form $\omega^{1+\alpha}$,
Schlicht and Stephan \cite{SchlichtS11} showed that the
supremum of the 
$\Lf$-automatic ordinals is exactly 
$\omega^{\omega^{\alpha+1}}$ whence Theorem
\ref{thm:BoundsonScatteredOrders} is optimal.
From our theorem we can also derive the following characterisation of
finite $\FC$-rank presentable ordinals (cf.~Appendix 
\ref{sec:OrdinalsOverFiniteRankOrders}).

\begin{corollary}
  Let $\Lf$ be a scattered linear order with
  $\FC(\Lf)<\omega$. The finite word $\Lf$-oracle-automatic ordinals
  are exactly 
  those below $\omega^{\omega^{\FC(\Lf)+1}}$. 
\end{corollary}
Here, the oracle is crucial: $0$ and $1$ 
are the only  finite word $\Z^n$-automatic ordinals if $n\geq 1$
(any $\Z^n$-automatic linear order with $2$ elements contains a
copy of $\Z$).

\subsection{Ranks of Well-Founded Automatic Order Forests}
\label{sec:Forests}

We next study scattered-oracle-automatic well-founded order forests. 
Kartzow et al.~\cite{KaLiLo12} proved compatibility of the ordinal height
with sum- and box-augmentations. Together with our decomposition
theorem this yields a bound on the height of an
$\Lf$-oracle-automatic well-founded order forest in terms of
$\FC(\Lf)$. Unfortunately, in important cases these bounds are not
optimal. For scattered orders $\Lf$ where the set of finite
$\Lf$-words allow an $\Lf$-oracle-automatic order which is scattered,
we can obtain better bounds. If $\Lf$ is an ordinal or has finite
$\FC$-rank, the set of
$\Lf$-words allows such a scattered ordering.
If the finite $\Lf$-words
admit an $\Lf$-automatic scattered order $\leq$, the
Kleene-Brouwer ordering of an $\Lf$-oracle-automatic well-founded
order forest with respect to $\leq$ is $\Lf$-oracle-automatic
again. Thus, its FC-rank is bounded 
by our previous result. Adapting a result of Kuske et
al.\cite{KuLiLo11} relating the $\FC$-rank of the Kleene-Brouwer
ordering with the height of the forest, we derive a bound on 
the height (cf.~Appendix \ref{app:RanksofForests}). Our main result on
forests is as follows. 

\begin{theorem}\label{thm:BoundsOnTreeRanks}
  \begin{itemize}
  \item Let $\Lf$ be an ordinal or a  scattered linear
    order with $\FC(\Lf)<\omega$.
    Each $\Lf$-oracle-automatic forest $\Ff=(F, \leq)$ 
    has rank strictly below $\omega^{\FC(\Lf)+1}$.
  \item Let $\Lf$ be some scattered linear order.
    Each $\Lf$-oracle-automatic forest $\Ff=(F, \leq)$ 
    has rank strictly below $\omega^{\omega\cdot (\FC(\Lf)+1)}$.
  \end{itemize}
\end{theorem}
\begin{remark}
  The bounds in the first part are optimal: for each ordinal $\Lf$ and
  each $c\in \N$, we can construct an $\Lf$-automatic tree of height
  $\omega^{\FC(\Lf)}\cdot c$ (cf.~Appendix \ref{sec:HighRanks}).
\end{remark}

\section{Separation of Tree- and Ordinal-Automatic Structures}
\label{sec:BooleanAlgebraNotAutomatic}

\begin{theorem}
  The countable atomless Boolean algebra is not finite word
  $\Lf$-automatic for any ordinal $\Lf$. 
\end{theorem}

This theorem is proved by first showing that, if the atomless Boolean
algebra is finite word $\Lf$-automatic for some ordinal $\Lf$, then it
already is $\omega^n$-automatic for some $n\in\N$. 
This follows because any finite word $\Lf$-automatic structure for $\Lf$ an ordinal above $\omega^{\omega}$ has a
sufficiently elementary substructure that has a $\omega^n$-automatic presentation for some $n\in\N$. 
In the case of the countable atomless Boolean algebra any $\Sigma_3$-elementary
substructure is isomorphic to the whole algebra. 
Extending Khoussainov et al.'s
monoid growth rate argument for automatic structures 
(cf.~\cite{DBLP:journals/lmcs/KhoussainovNRS07}) to the $\omega^n$-setting,
we can reject this assumption (cf. Appendix \ref{sec: Boolean algebra}). 
This answers a question of Frank Stephan.

\bibliography{ordinal-automata}
\bibliographystyle{plain}


\newpage
\appendix

\section{Basics on Scattered Linear Orders}

Recall the following basic (folklore) results. 

\begin{lemma} \label{lem:ClosedIntSmallerFCStar}
  Let $\Lf=(L, \leq)$ be a scattered linear order with
  $\FC(\Lf)=\alpha$.  For all $l,l'\in
  L$, there are some $n\in\N$ and scattered linear orders $\Lf_1,
  \Lf_2, \dots, \Lf_n$ of condensation rank strictly below $\alpha$ 
  such that $[l,l'] \cong  \Lf_1+ \Lf_2 + \dots + \Lf_n$. 
\end{lemma}
\begin{proof}
  $\Lf$ can be written as $\sum_{i\in\Z} \Lf_i$ for $\Lf_i$ scattered
  linear orders with $\FC(\Lf)<\alpha$. If $l$ comes from the $j$-th
  factor of this sum and $l'$ form the $j'$-th, then $[l,l']$ is
  isomorphic to $\Lf'_j + \sum_{i=j+1}^{j'-1} \Lf_i + \Lf'_{j'}$ where
  $\Lf'_j$ and $\Lf'_{j'}$ are  suborders of $\Lf_j$ and $\Lf_{j'}$
  whence they have rank below $\alpha$. 
\qed \end{proof}

\begin{lemma} \label{lem:FCstarSumsIncreaseRank}
  Let $\gamma\in\{\omega,\omega^*, \zeta\}$, 
  $\Lf_i$ be a scattered order of $\FCstar$ rank $\alpha$.
  The order $\Lf := \sum_{i\in\gamma} \Lf_i$ is of rank 
  $\FCstar(\Lf)=\alpha+1$. 
\end{lemma}
\begin{proof}
  Since $\FCstar(\Lf_i)=\alpha$, for all 
  $\beta<\alpha$ the $\beta$-th condensation of $\Lf_i$ contains
  infinitely many nodes. Thus, also the $\beta$-th condensation of
  $\Lf$ contains infinitely many equivalence classes containing
  elements in $\Lf_i$. 
  Thus, for each $i\in \gamma$ such that $i+2\in\gamma$ and
  for every $x_i\in\Lf_i$, $x_{i+2}\in\Lf_{i+2}$ 
  the $\beta$ condensation of $x_i$ and 
  the $\beta$-condensation of $x_{i+2}$ 
  are separated by infinitely many nodes (the $\beta$ condensations of
  $\Lf_{i+1}$). 
  Thus, the $\alpha$ condensation of $\Lf$ does not identify 
  nodes of $\Lf_i$ and $\Lf_{i+2}$. 
  Thus, it contains a suborder isomorphic to $\gamma$, whence 
  $\FCstar(\Lf)\geq \alpha+1$.
  On the other hand, since each $\Lf_i$ has rank $\alpha$ the
  $\alpha$-condensation of $\Lf$ is a $\gamma$-sum over finite linear
  orders. 
  Hence its $\alpha+1$-condensation is finite and
  $\FCstar(\Lf)\leq\alpha+1$. 
\qed \end{proof}

\begin{lemma}\label{lem:AllFCstarRankAsSubOrders}
  [Lemma 4.16 of \cite{Huschenbett12}]
  Let $\Lf$ be a linear order  and $\alpha<\FC(\Lf)$. 
  There is a closed interval $I$ of $\Lf$ such that
  $I$ is a scattered linear suborder of $\Lf$,
  $\FC(I) = \alpha+1$, and
  $\FCstar(I) = \alpha$. 
\end{lemma}

%

\section{Correctness of the Automaton in Example \ref{exa:FiniteWordsRec}}
\label{sec:FiniteWordsRecognisable}

  States 
  $e_l$ and $e_r$
  report errors from left and from right, respectively, i.e., a cut is
  forced to be visited in state $e_l$ if 
  it is a right limit step such that left of this limit infinitely many
  positive positions appear.
  
  On input $w$, the successor transitions mark the support of $w$ by
  state $p$ and all other successor positions in $w$ by $n$. 
  Let $P(w)\subseteq \Cuts(w)$ be defined by
  $(C,D)\in P(w)$ if $\exists x\in \supp(w)$ such that
  $x=\max(C)$. If $w$ has finite support, then 
  $$r(c):=
  \begin{cases}
    p &\text{if }c\in P(w)\\
    n & \text{otherwise}
  \end{cases}
  $$ defines an accepting run on $w$. 

  We now prove that there is no accepting run if $w$ is not a finite
  word. Heading for a contradiction assume that $r$ is an accepting
  run of $\Ac$ on $w$ and $w$ has infinite support. Then 
  $r((\emptyset, L)) = n = r((L, \emptyset))$. 
  We want to show that there is a cut $c$ such that $r(c)=e_l$ or
  $r(c)=e_r$. 

  If we are able to show this, we arrive at a
  contradiction: if $r(c)=e_l$ then $c$ is not the maximal cut. But
  there is no successor transition and no left limit transition from
  state $e_l$. Thus, $r$ cannot assign states to the cuts to the right
  of $c$, which is a contradiction. If state $e_r$ occurs, the
  argument is the same using the cuts to the left of $c$. 
  
  We show that there is a cut that is assigned an error state $e_l$ or
  $e_r$. 
  Assume that there is an infinite ascending chain
  $l_1<l_2<l_3<\dots$ in $\Lf$  such that 
  $\{l_i\mid i\in\N\}\subseteq \supp(w)$.
  For $C:=\{ x\in L\mid \exists i\in\N \quad x\leq l_i\}$ and
  $D:=L\setminus C$ the cut $c:=(C, D)$ has no direct
  predecessor. Moreover,  
  $p\in \lim_{c^-} r$ because state $p$ occurs at each cut associated
  to one of the $l_i$. Thus, if there is a right limit transition
  applicable at $c$, it assigns state $e_l$ to $c$. 
  If there is no infinite ascending chain in $\supp(w)$, then there is
  an infinite descending chain. The analogous argument shows that then
  state $e_r$ occurs.

\section{Proof of Theorem \ref{thm:BoundsonScatteredOrders}}
\label{sec:ProofOnScatteredOrders}

Huschenbett \cite{Huschenbett13} used the sum-of-box decomposition
technique in 
order to prove a strict bound on the finite condensation rank of
tree-automatic  scattered linear orders. His result relies on the fact
that the finite condensation rank behaves well with box-decompositions
in the following sense. Let $\alpha_0\oplus\dots\oplus\alpha_n$ denote
the natural sum (also known as commutative sum or Hessenberg sum) of
$\alpha_0,...,\alpha_n$.  

\begin{lemma}
  \label{lem:FCStarAndBox}
  [Proposition 4.11 in \cite{Huschenbett12}]
  For each  scattered linear order $\Af$ that is a tamely-colourable
  box-augmentation of 
  $\Bf_1, \dots, \Bf_n$, its rank is bounded by
  \begin{align*}
    \FCstar(\Af) \leq \FCstar(\Bf_1) \oplus \FCstar(\Bf_2) \oplus
    \dots \oplus \FCstar(\Bf_n).
      \end{align*}
\end{lemma}

Moreover, Khoussainov et al. have already shown that $\FCstar$ rank
behaves well with sum-augmentations.
\begin{lemma} 
  \label{lem:FCStarAndSum}
  [Proposition 4.4 in \cite{KhoussainovRS05}]
  For each  scattered linear order $\Af$ that is a sum-augmentation of
  $\Bf_1, \dots, \Bf_n$, its rank is determined by
  \begin{align*}
    \FCstar(\Af) = \max\{\FCstar(\Bf_1), \FCstar(\Bf_2),
    \dots, \FCstar(\Bf_n)\}.
  \end{align*}
\end{lemma}

\begin{proposition}
  Let $\alpha$ be a scattered order of $\FC$ rank
  $1+\gamma$ ($0$, respectively) for some ordinal $\gamma$. 
  Every $\alpha$-oracle-automatic scattered linear order has $\FCstar$
  rank strictly below $\omega^{\gamma+1}$ ($\omega^0=1$, respectively).
\end{proposition}
\begin{proof}
  In the case $\FC(\alpha)=0$ the domain of an $\alpha$-automatic structure
  has at most $\lvert \Sigma \rvert$ many elements. The theorem
  follows because every finite linear order has $\FCstar$ rank $0$. 

  Now let $\FC(\alpha)=1+\gamma$. As induction hypothesis assume that
  the theorem holds for all orders $\beta$ with $\FC(\beta)<1+\gamma$. 
  Heading for a contradiction assume that $\Lf=(L, \leq)$ is an
  $\alpha$-oracle-automatic scattered linear order such that
  $\FCstar(\Lf)\geq \omega^{\gamma+1}$. Let $\leq$ be recognised by
  some automaton with state set $Q_\leq$. 
  Due to  Lemma \ref{lem:AllFCstarRankAsSubOrders} the automaton $\Ac$
  corresponding to the 
  formula $\varphi(x,y_1,y_2) := y_1\leq x \leq y_2$ is a parameter
  automaton such that for each $n\in\N$ there is a parameter $p_n$
  such that $\Lf_{p_n}$ is a scattered linear order with
  $\FCstar(\Lf_{p_n}) = \omega^{\gamma}\cdot n$. Assume that $\Ac$ has
  state set $Q$.
  
  Now, fix some $n_0\in \N$ such that 
  $n_0 >  4^{2+2\cdot \exp(\lvert Q \rvert^2+2\lvert Q_\leq\rvert^2)}$.
  Due to Theorem \ref{thm:Decomposition}, there are sets
  $\Cc_0, \Cc_1$ of size $\exp(\lvert Q \rvert^2+2\lvert Q_\leq\rvert^2)$
  such that for each $n\leq n_0$, $\Lf_{p_n}$ is a tamely-colourable
  sum-of-box-augmentation of $\Cc_0, \Cc_1$ and 
  sets of $\beta_i$-oracle-automatic structures where
  $\FC(\beta_i)<\FC(\alpha)$ (cf. Lemma \ref{lem:ClosedIntSmallerFCStar}).
  By choice of $n_0$, there is some $1\leq m < \frac{n_0}{4}$ such that
  for all  structures  $\Af\in \Cc_0\cup \Cc_1$
  \begin{align}
    \label{eq:FCRankGaps}
    \begin{aligned}
      &\FCstar(\Af) \leq  \omega^{\gamma} \cdot m \text{ or}\\
      &\FCstar(\Af) > \omega^{\gamma} \cdot 4m.
    \end{aligned}
  \end{align}
  Now consider the decomposition of 
  $\Lf_{p_{4m}}$. Due to Lemma \ref{lem:FCStarAndSum} 
  there is a suborder $\Lf'$ of $\Lf_{p_{4m}}$  with
  $\FCstar(\Lf')=\omega^{\gamma} \cdot 4m$ that is tamely-colourable
  box-augmentation of structures
  $(\Cf_0, \Cf_1, \Bf_1, \dots, \Bf_k)$ where
  $\Cf_0\in \Cc_0, \Cf_1\in \Cc_1$, and $\Bf_i$ a
  $\beta_i$-oracle-automatic structure for each $1\leq i \leq k$. 
  Note that for each $1\leq i \leq k$, by induction hypothesis
  $\FCstar(\Bf_i) < \omega^{\gamma_i+1}$ for some $\gamma_i<\gamma$. 
  Thus,
  \begin{equation*}
    \FCstar(\Bf_1) \oplus
    \dots \oplus  \FCstar(\Bf_k) < \omega^{\max\{\gamma_i\mid 1\leq i
      \leq k\}+1}\leq \omega^{\gamma}.
  \end{equation*}
  Moreover, since $\Cf_0$ and $\Cf_1$ are substructures of $\Lf'$, 
  we have $\FCstar(\Cf_i)\leq \omega^\gamma \cdot 4m$ whence 
  \eqref{eq:FCRankGaps}
  implies that 
  $\FCstar(\Cf_i) \leq \omega^{\gamma}\cdot m$ for $i\in\{0,1\}$. 
  Due to the properties of $\oplus$ and Lemma \ref{lem:FCStarAndBox}
  we arrive at the contradiction 
  \begin{align*}
    \FCstar(\Lf')=\omega^{\gamma} \cdot 4m &\leq
    \FCstar(\Cf_0) \oplus \FCstar(\Cf_1) \oplus \FCstar(\Bf_1) \oplus
    \dots \oplus
    \FCstar(\Bf_k)\\  
    &\leq   \omega^\gamma\cdot m  \oplus \omega^\gamma \oplus
    \omega^\gamma \cdot m \\
    &<    \omega^{\gamma} \cdot 4m.
  \end{align*}
\qed \end{proof}

Theorem \ref{thm:BoundsonScatteredOrders} now follows as a corollary
of this Proposition.

\begin{corollary}
  Let $\alpha$ be a scattered order of $\FCstar$ rank
  $1+\gamma$ ($0$, respectively) for some ordinal $\gamma$. 
  Every finite word $\alpha$-oracle-automatic scattered linear order
  has $\FCstar$ 
  rank strictly below $\omega^{\gamma+1}$ ($\omega^0=1$, respectively).
\end{corollary}
\begin{proof}
  If $\alpha$ is a scattered linear order such that
  $\FCstar(\alpha)=1+\gamma$, then there are linear orders 
  $\alpha_i$  with 
  $\FC(\alpha_i)\leq 1+\gamma$
  for $1\leq i \leq k$ such that
  $\alpha=\sum_{i=1}^k \alpha_i$.
  
  Theorem \ref{thm:Decomposition} implies that each finite word
  $\alpha$-oracle-automatic scattered linear order $\Lf$ is a tamely
  colourable sum-of-box augmentations of 
  $(\Fc_{\alpha_1}, \dots, \Fc_{\alpha_k})$,  the classes of finite
  word $\alpha_i$-oracle-automatic structures. 
  Due to Lemmas \ref{lem:FCStarAndBox} and \ref{lem:FCStarAndSum}
  there are  $\alpha_i$-oracle-automatic scattered linear orders
  $\Lf_i$ (for $1\leq i \leq k$) such that
  $\FCstar(\Lf) \leq \FCstar(\Lf_1) \oplus \dots \oplus
  \FCstar(\Lf_k)$. Since $\FCstar(\Lf_i) < \omega^{\gamma+1}$ for each
  $1\leq i \leq k$, we immediately conclude that
  $\FCstar(\Lf) < \omega^{\gamma+1}$. 
\qed \end{proof}

\section{Ranks of Forests}
\label{app:RanksofForests}

We now introduce a variant of the height of a well-founded partial
order called \emph{infinity rank} and
denoted by $\infrank$.

\begin{definition}
  Let $\Pf=(P, \leq)$ be a well-founded partial order.
  We define the ordinal valued $\infrank$ of a node $p\in P$ inductively by
  \begin{align*}
    \infrank(p, \Pf)=\sup\{ \alpha+1\mid \exists^\infty p' ( p'<p \text{ and }
    \infrank(p', \Pf)\geq \alpha)\}.
  \end{align*}
  The $\infrank$ of $\Pf$ is then
  \begin{align*}
    \infrank(\Pf)=
    \sup\{\alpha+1\mid \exists^\infty p\in P\quad  
    \infrank(p, \Pf)\geq \alpha\}. 
  \end{align*}
\end{definition}

\begin{lemma}\label{lem:RankAndInfRank}
  \cite{KaLiLo12}
  For $\Pf$  a well-founded partial order,
  we have 
  \begin{equation*}
    \infrank(\Pf) \leq \rank(\Pf)< \omega\cdot (\infrank(\Pf) +1).    
  \end{equation*}
\end{lemma}

In this section, we prove the following bound on the ranks of
$\alpha$-automatic order forests.
\begin{theorem}\label{thm:BoundsOnTreeRanksDet}
  Let   $\alpha$ be some scattered linear order.
  \begin{enumerate}
  \item 
    Every $\alpha$-oracle-automatic order forest $\Ff=(F, \leq)$ such
    that
    \begin{itemize}
    \item $F$ is also the domain of some $\alpha$-oracle-automatic
      scattered linear order, and
    \item    $\FC(\alpha)= 1+\gamma$
    \end{itemize}
    has rank
    strictly below $\omega^{1+\gamma+1}$ and
    $\infrank$ strictly below $\omega^{\gamma+1}$.
  \item 
    \begin{itemize}
    \item 
    If $\FC(\alpha) <\omega$, then 
    every $\alpha$-oracle-automatic order forest has rank
    strictly below $\omega^{ \FC(\alpha)+1}$.    
  \item     
    If $\FC(\alpha) = \omega +c_0$ for some $c_0<\omega$, then 
    every $\alpha$-oracle-automatic order forest has rank
    strictly below $\omega^{\omega \cdot (c_0 + 1)}$.    
  \item 
    If $\FC(\alpha) = \omega \cdot c_1 +c_0$ for $c_0,c_1<\omega$ and
    $c_1\geq 2$, then  every $\alpha$-oracle-automatic order forest has rank
    strictly below $\omega^{\omega^2\cdot (c_1-1) + \omega  \cdot (c_0
      + 1)}$.    
  \item If $\FC(\alpha) \geq \omega^2$, 
    then every $\alpha$-oracle-automatic order forest has rank
    strictly below $\omega^{ \omega \cdot \FC(\alpha)+
      \omega}$.\footnote{
      In particular, 
      if $\FC(\alpha) = 
      \omega^n \cdot c_n +
      \omega^{n-1} \cdot c_{n-1} +
      \dots +
      \omega \cdot c_{1} +
      c_0$ such that $n\geq 2$, $c_1,c_2,\dots, c_n <\omega$, and
      $c_n\neq 0$,
      then every $\alpha$-oracle-automatic order forest has rank
      strictly below $\omega^{\omega^{n+1} \cdot c_n +
        \omega^{n-1+1} \cdot c_{n-1} \dots \omega^2\cdot c_1 + \omega
        \cdot( c_0+1) }$.
    }
  \end{itemize}
  \end{enumerate}  
\end{theorem}
\begin{remark}
  If $\alpha$ is an ordinal or $\FC(\alpha)<\omega$, 
  we show in the next section that  every  $\alpha$-oracle-automatic  
  set $F$ of finite $\alpha$-words allows a scattered linear order.
  Thus, if $\alpha$ satisfies one of these conditions, then the better
  bounds hold.
\end{remark}

\subsection{A Scattered Order of Scattered Words}
\label{sec:OrderingWords}

We first show that scattered orders $\alpha$ of finite rank allow a
scattered order of all finite $\alpha$-words that is
$\alpha$-automatic. Afterwards, 
we show that the analogous result holds in case that $\alpha$ is an
ordinal. 
Our first claim is proved by induction on the $\FC$-rank and the
$\FCstar$-rank of $\alpha$. We prepare our result by defining an
automaton that determines at every cut the left and the right rank of
this cut. Given a cut $c=(C,D)$ without direct predecessor, the left rank is
the minimal rank of the induced suborders of nonempty upwards closed
subsets of $C$. Analogously, the right rank is the minimal rank of the
induced suborders of nonempty downwards closed subsets of $D$.

\begin{definition} \label{def:AutLimitStages}
  For $\Sigma$ arbitrary, 
  let $\Cc_n=(Q_n, \Sigma, I_n,F_n,\Delta_n)$ be an automaton with state set 
  $Q_n:=\{0, 1, \dots, n\} \times \{0, 1, \dots, n\}$, 
  initial states $I_n=\{0\}\times \{0, 1, \dots, n\}$ and
  final state $F_n=\{0, 1, \dots, n\} \times \{0\}$.
  In order to define its transition relation, we use the following
  notation for $i\leq n$, let
  $\Pc_i$ be defined by
  \begin{align*}    
    \{ S\in 2^{Q_n}\mid 
    \forall j> i\ \forall k\quad (j,k),(k,j) \notin S \text{ and }
    \exists k\leq i\quad
    (i,k)\in S \text{ or } (k,i)\in S\}.
  \end{align*}
  The transition relation of $\Cc_n$ is
  \begin{align*}
    \Delta_n =& \left\{( (i,0),\sigma,(0,j))\mid \sigma\in\Sigma\text{
        and }i,j\in\{0,
      1, \dots, n\} \right\}\\
    &\cup \{( (i,j), X)\mid X\in \Pc_j\}\\
    &\cup \{( X, (i,j))\mid X\in \Pc_i\}
  \end{align*}
\end{definition}

\begin{lemma}
  Let $\alpha$ be some scattered linear order and $w$ an arbitrary
  $\alpha$-word. 
  Interpreting $\Cc_n$ as an $\alpha$-automaton, there is an accepting
  run $r$ of $\Cc_n$ on $w$ if and only if $\FCstar(\alpha)\leq n$.
  In this case, $r$ is the unique accepting run and for every cut 
  $c=(C,D)$ the state at $c$ is
  \begin{itemize}
  \item in $\{0\}\times\{0, 1, \dots, n\}$ if $c$ has a direct predecessor,
  \item in $\{0, 1, \dots, n\}\times \{0\} $ if $c$ has a direct successor,
  \item in $\{k\} \times \{0, 1, \dots, n\}$ (with $k\geq 1$) if
    $c$ has no direct predecessor, and for each cut $c'<c$ there is a cut
    $c''$ such that $c'<c''<c$  and
    $\FC(\alpha{\restriction}_{(c'',c)})=k$, and
  \item in $\{0, 1, \dots, n\}\times \{k\}$ (with $k\geq 1$) if
    $c$ has no direct successor, and for each cut $c'>c$ there is a cut
    $c''$ such that $c'>c''>c$  and
    $\FC(\alpha{\restriction}_{(c,c'')})=k$. 
  \end{itemize}
\end{lemma}
\begin{proof}
  First, let $n \geq \FCstar(\alpha)$. This implies, that for all cuts
  $c''$ and $c$, the suborder induced by $(c'',c)$ has $\FC$-rank at
  most $\FCstar(\alpha)\leq n$. 
  Moreover, if $c$ is a cut without direct predecessor, and if
  $c_1<c_2<c_3<\dots    <c$ is an infinite chain of cuts whose limit is $c$,
  then  $\FC(\alpha{\restriction}_{(c_i,c)})$ stabilises at some
  $i_0$. Thus, the following function $r$ is well-defined. It is a function
  $r: \Cuts(\alpha)\to Q_n$ where for each cut $c=(C,D)$ we have
  $r(C,D)= (i,j)$ such that
  \begin{enumerate}
  \item $i=0$ if $c$ has a direct predecessor or $C=\emptyset$, 
  \item  otherwise, $i= \min\{ \FC( (c',c))\mid c'<c\}$,
  \item $j=0$ if $c$ has a direct successor or $D=\emptyset$, 
  \item  otherwise, $j= \min\{ \FC( (c,c'))\mid c'>c\}$.
  \end{enumerate}
  A straightforward induction on the left and right rank of each cut
  in $\alpha$ shows that $r$ is consistent with the transition
  relation, i.e.,  $r$ is an accepting run of $\Cc_n$ on each
  $\alpha$-word.  
  
  We next show that $r$ is the unique run  of $\Cc_n$ on
  $\alpha$-words.
  Heading for a contradiction assume that $r'$ is another accepting
  run on some $\alpha$-word and that $c=(C,D)$ satisfies
  $r(c) = (i,j) \neq r'(c)=(i',j')$. Without loss of generality (the
  other case is symmetric), we
  may assume that $i\neq i'$ and  $c$ has been chosen such that $i$ is
  minimal with this property. 
  We distinguish the following cases:
  \begin{itemize}
  \item Assume that $i=0$. Since $r'$ is accepting, $c$ cannot be the minimal
    cut. Thus, $c$ has a direct predecessor $c'$. But independent of
    the
    successor transition used between $c'$ and $c$, $r'(c)\in
    \{0\}\times\{0, 1, \dots, n\}$ whence $i=i'=0$ contradicting the
    assumption $i\neq i'$.
  \item Assume that $i\geq 1$. The right limit transition applied by
    $r$ at $c$ 
    shows that there is a cut $c'<c$ such that for all
    $c''\in (c',c)$, $r(c'')\in \{0, 1, \dots, i-1\}^2$. 
    By minimality of $i$, $r$ and $r'$ agree on this interval. 
    But then again the applicable right limit transitions always
    imply that $i'=i$ contradicting $i'\neq i$. 
  \end{itemize}
  
  Finally, we have to show that there are no accepting runs of $\Cc_n$
  on $\alpha$-words if $\FCstar(\alpha)>n$. 
  Assume that $\FCstar(\alpha)>n$. Due to Lemma
  \ref{lem:AllFCstarRankAsSubOrders}, $\alpha$ contains an interval
  $\alpha'$ with $\FCstar(\alpha')=n+1$. We show that there is no
  function $r:\alpha'\to Q_n$ which is consistent with the transition
  relation $\Delta_n$. 
  Up to symmetry, $\alpha'$ contains an upwards closed interval of the
  form  $\sum_\omega \beta_i$  with $\FC(\beta_i)=n$. 
  As shown in the first part, there is an accepting run $r'$ of
  $\Cc_{n+1}$ on this sum. For the maximal cut $c_{\max}$ of $\alpha'$,
  we have $r'(c_{\max})= (n+1,0)$. In fact, one easily sees that the
  previous arguments apply to any (possibly non-accepting run) on
  $\alpha'$ in the sense that any run on $\alpha'$ satisfies
  $r'(c_{\max})\in \{n+1\}\times \{0, 1, \dots, n+1\}$. 
  Since $\Delta_n\subsetneq \Delta_{n+1}$, any run of $\Cc_n$ on
  $\alpha$ is also a run of $\Cc_{n+1}$ that does not use states from
  $\{n+1\}\times\{0, 1, \dots, n+1\}$. But we have seen that any run
  of $\Cc_{n+1}$ on $\alpha'$ would label $c_{\max}$ with such a
  state. Thus, there is no run of $\Cc_n$ on $\alpha'$ whence there
  can neither be a run of $\Cc_n$ on $\alpha$. 
\qed \end{proof}

The automaton $\Cc_n$ will be useful to decompose an order $\alpha$
with $\FCstar(\alpha)=n$ into finitely many pieces
$\alpha=\alpha_1+\alpha_2+\dots+ \alpha_k$ of $\FC$-rank at most $n$. 

\begin{lemma}
  Let $\alpha$ be an order with $\FCstar(\alpha)=n$ and $r$ the
  accepting run of $\Cc_n$ on $\alpha$-words. Let
  $c,d$ be consecutive cuts of maximal rank in the sense that 
  \begin{itemize}
  \item $c$ is minimal or $r(c)=(i,j)$ with
    $\max(i,j)=n$,
  \item $d$ is maximal or $r(d)=(k,l)$ with $\max(k,l)=n$, and
  \item for all $e\in (c,d)$, $r(e)=(x,y)$ we have $\max(x,y)<n$.
  \end{itemize}
  Then the interval $(c,d)$ of $\alpha$ has $\FC$-rank at most $n$. 
\end{lemma}
\begin{remark}
  In particular, this lemma implies that in an order $\alpha$ with
  $\FCstar(\alpha)=n$ there are only finitely many
  cuts of left or right rank $n$.
\end{remark}
\begin{proof}
  By induction on $i$, we  prove that
  for arbitrary cuts $c \leq d$ the following holds.
  If for all cuts $e$ strictly between $c$
  and $d$ we have  $r(e)\in \{0, 1, \dots, i-1\}^2$ then
  $\FC( (c,d))\leq i$. 
  
  For $i=0$, the condition implies that $c=d$ whence
  $\FC( (c,d) )=\FC(\emptyset) = 0$. 
  Now assume that this claim holds for $i-1$ and that
  for all cuts $e\in (c,d)$ we have
  $r(e)\in \{0, 1, \dots, i-1\}^2$. 
  By definition of the limit transitions, we know that
  $r(c)\in \{0, 1, \dots, n\}\times \{0, 1, \dots, i\}$ and that
  $r(d)\in \{0, 1, \dots, i\}\times \{0, 1, \dots, n\}$. 
  From our construction of the accepting run $r$ (compare the previous
  proof), we conclude that there are cuts
  $c<c_1\leq d_1< d$ such that
  $\FC((c,c_1)) \leq i-1$ and
  $\FC((d_1,d))\leq i-1$. 
  Next, we claim that there are only finitely many cuts $c_1 < e <
  d_1$ such that $r(e)\in  M_{i-1}:=
  (\{i-1\}\times \{0, 1, \dots, \{i-1\}) \cup
  (\{0, 1, \dots, \{i-1\}) \times \{i-1\})$.
  Otherwise there would be an infinite ascending or descending chain
  of cuts in $M_{i-1}$ whose limit $e$ would satisfy $c_1\leq e \leq
  d_1$ and $r(e)\notin \{0, 1, \dots, i-1\}^2$ contradicting our
  assumptions on the interval $(c,d)$. 
  Thus, let $c_1=e_1 < e_2 < \dots < e_{n-1} < e_n=d_1$ be a finite
  sequence of cuts such that for all $c_1\leq e \leq d_1$ we have
  $r(e)\in M_{i-1}$ only if there is a $1\leq j \leq n$ with $e=e_j$.
  Thus, $(c,d) = (c,c_1) + \sum_{i=1}^{n-1} (e_i,e_{i+1}) + (d_1,d)$
  is a finite sum of intervals that (by induction hypothesis) have
  $\FC$-rank at most $i-1$. Thus, $\FC((c,d))\leq i$ as desired.
\qed \end{proof}

Let us collect one more fact about $\Cc_{n+1}$. 
Assume that $\alpha$ is an order with
$\FC(\alpha)=\FCstar(\alpha)=n+1$. 
This implies that $\alpha = \sum_{\gamma\in\Gamma} \alpha_\gamma$
where
$\Gamma\in\{\omega, \omega^*, \Z\}$ and $\FC(\alpha_\gamma)\leq n$
where for infinitely many $\gamma\in\Gamma$ we have
$\FC(\alpha_\gamma)=n$. Thus,  $\Cc_n$ has an accepting run on each
$\alpha_\gamma$ that agrees with the run of $\Cc_{n+1}$ on $\alpha$ on
the interval $\alpha_\gamma$. Hence, the run of $\Cc_{n+1}$ assumes
only finitely many often a state from $M_n:=\{n\}\times\{1, 2, \dots,
n\} \cup \{1, 2, \dots, n\} \times n\}$ on each $\alpha_\gamma$. 
The next lemma follows immediately.

\begin{lemma}\label{lem:MaxRankCutsEmbedinZ}
  Let $\alpha$ be an order with
  $\FC(\alpha)=n+1$. Let $r$ be the accepting run of
  $\Cc_{n+1}$ on some $\alpha$-word. 
  The suborder induced by the cuts
  $\{c\mid r(c)\in M_n\}$ form a suborder of $\Z$. 
\end{lemma}

Thus, there is an accepting run of
$\Cc_{n+1}$ on every $\alpha$-word but no run of
$\Cc_n$ on some $\alpha$-word.

\begin{lemma} \label{lem:alpha-autmatic-words-are-scattered-ordered}
  Suppose that $\alpha$ is a scattered linear order with
  $\FC(\alpha)<\omega$. Then there is an $\alpha$-oracle-automatic
  scattered linear order on the 
  set of finite $\alpha$-words.   
\end{lemma} 

\begin{proof} 
  We define automata $\Ac_n$ (and $\Bc_n$, respectively) for each
  $n<\omega$ which 
  uniformly define 
  $\alpha$-automatic scattered linear orders on the finite
  $\alpha$-words over a fixed alphabet $\Sigma$ for all 
  scattered linear orders $\alpha$ with $\FC(\alpha)\leq n$ (and
  $\FCstar(\alpha)\leq n$,
  respectively).  
  Note that for $\alpha$ with $\FC(\alpha)=0$ there is a finite number
  of $\alpha$-words over 
  $\Sigma$ whence the construction of $\Ac_0$ is trivial.
  
  Suppose that we have constructed $\Ac_n$. We define $\Bc_n$ as follows. If
  $\FCstar(\alpha)\leq n$, 
  the run of the automaton $\Cc_n$ partitions $\alpha$ uniquely into a
  finite sum of intervals $\alpha=\alpha_1+\alpha_2+\dots+ \alpha_m$
  of $\FC$-rank $\leq n$ by taking the states from $\{n\}\times\{0, 1,
  \dots, n\} \cup \{0, 1, \dots, n\} \times\{n\}$ as splitting
  points. 
  Then 
  $\Bc_n$ orders $\alpha$-words lexicographically by comparing the restrictions
  to the intervals $\alpha_i$ via $\Ac_n$. 
  If $\Ac_n$ orders $\alpha_i$-words as some order $L_i$, then 
  $\Bc_n$ orders $\alpha$-words as the scattered sum 
  $\sum_{a_1\in L_1} \sum_{a_2\in L_2} \dots \sum_{a_m\in L_m} 1$ of
  one element orders which clearly is scattered again. 
  
  Suppose that $\FC(\alpha)\leq n+1$. Let $r$ be the accepting run of
  $\Cc_{n+1}$ on every 
  $\alpha$-word. 
  Recall that from Lemma \ref{lem:MaxRankCutsEmbedinZ}, we
  conclude that the cuts of rank $n$ embed into $\Z$, Thus, the cuts
  \begin{align*}
    C:=\{c \mid c \text{ minimal or maximal or } r(c)\in M_n\}
  \end{align*}
  are a suborder of $1+\Z+1$. 
  Given an $\alpha$-word $w$ we define $c(w)$ to be maximal element
  $c\in C$ such that $c<\supp(w)$ and define $d(w)$ to be the minimal
  element $c\in C$ such that $\supp(w)<c$. 
  We define $\Ac_{n+1}$ as
  follows.
  Given finite $\alpha$-words
  $v,w$, let $v\leq w$ if
  \begin{enumerate}
  \item $c(v) < c(w)$ , or
  \item $c(v) = c(w)$ and $d(v) < d(w)$, or
  \item $c(v) = c(w)$, $d(v)=d(w)$ and
    $\Bc_n$ applied to the interval between $c(v)$ and $d(v)$ reports 
    $v<w$. Note that $\FCstar((c(v),d(v))\leq n$ because the accepting
    run of $\Cc_{n+1}$ on $\alpha$ assumes only finitely many states
    of rank $n$ on this subinterval. Thus, also $\Cc_n$ accepts
    $(c(v),d(v))$-words. 
  \end{enumerate}
  This defines an $\alpha$-automatic linear order $(W_\alpha,\preceq)$
  on the set of finite $\alpha$-words $W_\alpha$. Since $\preceq$
  embeds into a $(1+\Z+1)^2$-sum  of scattered linear orders
  (induced by $\Bc_n$), where $(1+\Z+1)^2$ is ordered
  lexicographically, $(W_\alpha,\preceq)$   is scattered.  
\qed \end{proof} 

\begin{lemma}
  Let $\alpha$ be some ordinal. Then there is an $\alpha$-automatic
  well-order of all finite $\alpha$-words over an alphabet $\Sigma$.
\end{lemma}
\begin{proof}
  Fix a linear order $\leq_\Sigma$ on $\Sigma$.
  Let $w,v$ be $\alpha$-words. We set
  $w<v$ if either $\max(\supp(w))<\max(\supp(v))$ or
  $\max(\supp(w)) = \max(\supp(v))$ and
  there is a $\beta<\max(\supp(w))$ such that 
  $w(\beta) <_\Sigma v(\beta)$ and for all $\alpha>\beta'>\beta$,
  $w(\beta')=v(\beta')$. 
  Apparently this order is $\alpha$-automatic. 
  Note that for $\alpha=\omega$ this is a the
  length-backward-lexicographic order (we first compare words with
  respect to size and words of the same size are compared
  lexicographically from the last letter to the first one). 
  In order to show that this defines a well-order, first note that it
  is reflexive, transitive and antisymmetric, i.e., a linear order.
  Heading for a contradiction, assume that there is some ordinal
  $\alpha$ such that the order on $\alpha$-words contains an infinite
  descending chain $w_1 > w_2 > w_3 > \dots$.
  The chain $\alpha_i:=\max(\supp(w_i))$ is a monotone
  decreasing sequence in $\alpha$. Since $\alpha$ is an ordinal, it
  stabilises at some $k\in\N$. We conclude that
  the sequence $v_j:=w_{k+j}$ satisfies
  $\max(\supp(v_{j}))= \alpha_k$ for all $j\in\N$.

  We now iterate the following argument: let $\alpha'<\alpha_k$ be
  maximal such that there are $v_j, v_k$ such that $v_j(\alpha') \neq
  v_k(\alpha')$. Since $\Sigma$ is finite, there is an infinite
  subsequence $v_{i_1} > v_{i_2} > \dots$ such that
  $v_{i_k}$ and $v_{i_j}$ agree at $\alpha'$, i.e.,
  $v_{i_k}(\alpha')=v_{i_j}(\alpha')$. Replace the sequence $v_{k}$ by
  the sequence $v_{i_k}$. Since this is an decreasing chain and above
  $\alpha'$ all $v_{i_k}$ agree, we can repeat this argument with some
  smaller $\alpha''<\alpha'$ which is maximal such that some $v_{i_k}$
  do not agree on $\alpha''$. 
  Since $\alpha$ is an ordinal and since
  $\alpha'>\alpha''>\alpha'''>\dots$, this sequence must be finite. 
  But this process terminates if and only if $v_{i_k} = v_{i_j}$ for
  all $j,k\in\N$. This contradicts the assumption that the $v_{i_j}$
  form a strictly decreasing infinite chain.
\qed \end{proof}

\subsection{Kleene-Brouwer Orders of Trees}
\label{sec:KleeneBrouwer}

Let $\Tf=(T, \sqsubseteq)$ be a tree and let $\Lf=(T, \preceq)$ be a
linear order. Then we can define
the  \emph{Kleene-Brouwer order} (also called Lusin-Sierpi\'nski
order)  $\KB(\Tf, \Lf):=(T, \lessdot)$ given by
$t \lessdot t'$ if either 
$t \sqsubseteq t'$ or there are $t\sqsubseteq s, t'\sqsubseteq s'$
such that  
$\{r\in T\mid s\sqsubset r\} = \{r\in T\mid s'\sqsubset r\}$ and 
$s\prec s'$. This
generalises the order induced by postorder traversal to 
infinitely branching trees where the children of each node are ordered
corresponding to the linear order $\preceq$. 
Since $\alpha$-oracle-automatic structures are closed under
first-order definitions, the following observation is immediate.
\begin{proposition}
  If $\Tf$ is an tree and $\Lf$ a linear order such that both are
  $\alpha$-oracle-automatic, then $\KB(\Tf, \Lf)$ is
  $\alpha$-oracle-automatic. 
\end{proposition}

For the following section, it is important that
$(T, \lessdot)$ is scattered if
$(T, \preceq)$ is a scattered linear order. 

\begin{lemma}
  Let $\Tf=(T, \sqsubseteq)$ be a tree and $\Lf=(T, \preceq)$ a
  scattered linear order, then $\KB(\Tf, \Lf)=(T, \lessdot)$ is
  scattered.
\end{lemma}
\begin{proof}
  The proof is by induction on the rank of $\Tf$. If $\Tf$ has rank
  $1$, it consists only of the root whence $\KB(\Tf, \Lf)$ is the linear
  order of $1$ element which is scattered. 
  Otherwise, let $T_0$ be the set of children of the root and let
  $t_0$ be the root of $\Tf$. 
  $T_0$ induces a scattered suborder $(T_0, \preceq)$ of $\Lf$. 
  Now (abusing notation slightly) 
  $\KB(\Tf, \Lf) = \left(\sum_{t\in (T_0, \preceq)} \KB(\Tf(t), \Lf)\right) +
  t_0$ which is a scattered sum of 
  scattered orders. Proposition 2.17 in \cite{rosenstein82} shows that
  $\KB(\Tf, \Lf)$ is scattered. 
\qed \end{proof}

\subsection{Bounds for Forests on Scattered Orders of Finite Rank}
\label{sec:Treeboundsmall}

In this section, we prove the main theorem in the case that
$\FC(\alpha)$ is finite, $\alpha$ is an ordinal, or in general, the
set of finite $\alpha$-words allows a scattered linear order. In the next
Section we then prove the other cases.

\begin{lemma}\label{lem:rankAndKB}
  Let $\Tf=$ be a nonempty $\alpha$-oracle-automatic order tree with domain
  $T$ and $\Lf$ a scattered $\alpha$-oracle-automatic order with
  domain $T$. 
  If $\FCstar(\KB(\Tf, \Lf)) < \beta$,
  then $\infrank(\Tf)<\beta$. 
\end{lemma}
\begin{proof}
  The proof is by contraposition and induction on $\beta$. 
  \begin{itemize}
  \item If $\beta=0$, there is nothing to show.
  \item Assume that $\infrank(\Tf)=\beta=\beta'+1$ and for each tree
    $\Tf'$ with $\infrank(\Tf')=\beta'$ we have
    $\FCstar(\KB(\Tf', \Lf)) \geq \beta'$. 
    By definition of $\infrank(\Tf)$ there is an infinite antichain
    $d_1, d_2, d_3,\dots$ in $\Tf$ such that the subtree
    $\Tf(d_i)$ rooted at $d_i$ satisfies
    $\infrank(\Tf(d_i))=\beta'$. By induction hypothesis, 
    $\FCstar( \KB(\Tf(d_i), \Lf))\geq \beta'$. 
    Moreover, $\Lf$ orders $\{d_i\mid i\in\N\}$ as order type 
    $\gamma\in\{\omega, \omega^*, \zeta\}$
    Thus, $\KB(\Tf, \Lf)$ contains a suborder of the form
    $\sum_{x\in \gamma} \KB(\Tf(d_x), \Lf))$ with
    $\FCstar(\KB(\Tf(d_x), \Lf))=\beta'$. 
    Due to Lemma \ref{lem:FCstarSumsIncreaseRank}, 
    we conclude that
    \begin{align*}
      \FCstar(\KB(\Tf, \Lf)) \geq 
      \FCstar(\sum_{x\in \gamma} \KB(\Tf(d_x), \Lf))       
      = \beta'+1 = \beta.      
    \end{align*}
  \item Assume that $\infrank(\Tf)=\beta$ is a limit ordinal. 
    By definition for each $\beta'<\beta$ there is $d\in \Tf$ such
    that
    $\infrank(\Tf(d)) \geq \beta'$ whence
    $\FCstar(\KB(\Tf(d), \Lf )) \geq \beta'$ by induction. 
    Thus, $\FCstar(\KB(\Tf, \Lf)) \geq \sup\{\beta' \mid \beta'<\beta\} =
    \beta$. 
  \end{itemize}
\qed \end{proof}

\begin{corollary}
  \label{cor:rankAndKB}
  Let $\Tf$ be a nonempty $\alpha$-automatic order tree. 
  If $\FC(\KB(\Tf, \Lf)) < \beta$,
  then $\infrank(\Tf)<\beta+1$.   
\end{corollary}

Combining this result with our bound on the $\FC$ ranks of
$\alpha$-oracle-automatic we can now prove the first part of 
Theorem \ref{thm:BoundsOnTreeRanksDet}.

\begin{proof}[Proof of Theorem \ref{thm:BoundsOnTreeRanksDet} part (1)] 
  Assume that  $\Tf=(T,\leq)$ is an
  $\alpha$-oracle-automatic 
  order tree  such that
  $\Lf$ is an $\alpha$-oracle-automatic scattered order with domain $T$. 
  Since $\KB(\Tf, \Lf)$ is an 
  $\alpha$-oracle-automatic scattered linear order,
  $\FC(\KB(\Tf, \Lf)) <\omega^{\gamma+1}$ due to Theorem 
  \ref{thm:BoundsonScatteredOrders}. 
  Due to Corollary \ref{cor:rankAndKB}, $\infrank(\Tf)<
  \omega^{\gamma+1}$.  
  By application of Lemma \ref{lem:RankAndInfRank} we finally obtain
  $\rank(\Tf)<\omega^{1+\gamma+1}$. 

  Note that this result easily extends to forests because for each
  $\alpha$-oracle-automatic forest, we can turn it into a
  $\alpha$-oracle-automatic tree by
  adding a new root. This tree has the same $\infrank$ as the forest
  we started with. 
\qed \end{proof}

\subsection{Bounds for Forests on Scattered Orders 
  of Infinite Rank} 
\label{sec:Treeboundlarge}

Since we do not know whether there is an $\alpha$-automatic scattered
linear ordering of all finite $\alpha$-words for all linear orders
$\alpha$  with $\FC(\alpha)\geq \omega$, we have to do a direct analysis of
the sum-of-box decompositions of $\alpha$-automatic forests. 
Fortunately, we can rely on the analogous analysis in the case of
tree-automatic structures from \cite{KaLiLo12}.
The essence of this analysis can be rewritten as the following result.
\begin{theorem} \label{thm:treerank-indecomposable}
  \cite{KaLiLo12}
  If $\Ff$ is a forest that is tamely-colourable
  sum-of-box-augmentation of classes $\Cc_1, \dots, \Cc_k$ such that
  for all 
  structures $\Ff'\in \bigcup_{i=1}^k \Cc_1$ we have  
  $\infrank(\Ff')\neq \omega^\alpha$, 
  then $\infrank(\Ff) \neq \omega^\alpha$. 
\end{theorem}

Using this decomposition result, the second part of
Theorem \ref{thm:BoundsOnTreeRanksDet} is obtained by induction.

\begin{proof}[Proof of Theorem \ref{thm:BoundsOnTreeRanksDet} part (2)] 
  Because of the first part of this theorem and Lemma 
  \ref{lem:alpha-autmatic-words-are-scattered-ordered}, the claim for
  orders $\alpha$ with $\FC(\alpha)<\omega$ has already been proved. 
  
  We now establish the following claim.
  Assume that $\alpha$ is a scattered linear order of rank
  $\FC(\alpha)=\gamma \geq \omega$.   
  Let $\delta \geq \omega $ be an ordinal such that for all $\alpha'$ with
  $\FC(\alpha')< \gamma$ and all 
  $\alpha'$-oracle-automatic forests $\Ff'$, 
  $\infrank(\Ff')<\omega^\delta$. 
  Then every $\alpha$-oracle-automatic forest $\Ff$ satisfies
  $\infrank(\Ff) < \omega^{\delta + \omega}$.
    
  Heading for a contradiction assume that $\Ff$ is an
  $\alpha$-oracle-automatic forest with 
  \mbox{$\infrank(\Ff)\geq \omega^{\delta + \omega}$.}
  Then there is a parameter automaton $\Ac$ (corresponding to the
  formula $x<y$ and parameters $p_n$ for $n\in\N$ such that
  $\infrank(\Ff_{p_n}) = \omega^{\delta + n}$. 
  Assume that $\Ac$ has $q$ many states and the order automaton of
  $\Ff$ has $q_<$ many states. 
  Now fix $n_0 > 2 \exp(q^2+2q_<^2)$. 
  Due to Theorem \ref{thm:Decomposition}, 
  there are sets
  $\Cc_0, \Cc_1$ of size $\exp(q^2+2q_<^2)$
  such that for each $n\leq n_0$, $\Ff_{p_n}$ is a tamely-colourable
  sum-of-box augmentation of $\Cc_0, \Cc_1$ and 
  some sets  of $\alpha_i$-oracle-automatic structures where
  $\FC(\alpha_i)< \gamma$ for each $i$.
  By choice of $n_0$, there is some $1 \leq m \leq n_0$ such that
  \begin{equation*}
     \infrank(\Af) \neq \omega^{\delta + m}
  \end{equation*}
  for all  structures  $\Af\in \Cc_0\cup \Cc_1$.
  Moreover, by definition of $\delta$ every $\alpha_i$-oracle-automatic
  forest has  $\infrank$ strictly below $\omega^\delta$. 
  Thus, $\Ff_{p_{m}}$ is a tamely-colourable sum-of-box augmentation 
  of classes of structures such that none of these
  structures has $\infrank$  $\omega^{\delta + m}$.
  But this contradicts directly  Theorem
  \ref{thm:treerank-indecomposable} because
  $\infrank(\Ff_{p_{m}} ) = \omega^{\delta+m}$.

  Using Lemma \ref{lem:RankAndInfRank} this claim carries 
  over from $\infrank$ to rank because 
  for
  $\gamma\geq \omega$ some forest has 
  rank strictly below $\omega^\gamma$ if and only if
  it has $\infrank$ strictly below $\omega^\gamma$ (note that
  $\omega\cdot \omega^\gamma = \omega^\gamma$).
  
  The proof of the theorem now follows by a straightforward induction
  on $\FC(\alpha)$ using the claim proved above. 
\qed \end{proof}

\subsection{Optimality of the Bounds on Forests}
\label{sec:HighRanks}

The upper bounds on the ranks of trees stated in the first part of
Theorem \ref{thm:BoundsOnTreeRanksDet} are optimal in the sense that we
can reach all lower ranks as stated in the following theorem. 

\begin{theorem}\label{thm:TreesLowerBounds}
  \begin{enumerate}
  \item For all  $i,c\in\N$ there is an $\omega^i$-automatic tree
    $\Tf_{i,c}$ with $\rank(\Tf_{i,c})=\omega^{i}\cdot c$.
  \item 
    For all ordinals $\gamma\geq\omega$ and all $c\in\N$, 
    there is an
    $\omega^{1+\gamma}$-automatic tree $\Tf_{\gamma,c}$ with
    $\rank(\Tf_{\gamma,c})= \omega^\gamma \cdot c$. 
  \end{enumerate}
\end{theorem}

In order to prove the first part of Theorem
\ref{thm:TreesLowerBounds},  we
want to  construct  for all $i\in\N$ and $c\in\N$ an
$\omega^i$-automatic tree of $\infrank$  $\omega^{i-1}\cdot c$ and
rank $\omega^{i}\cdot c$. 

We define a finite word
$\omega$-automatic tree as follows. Let $T=(\{\varepsilon\} \cup
\{ (n,m) \mid n\leq m \}$ and $\Tf_0=(T, \leq)$ where
\begin{align*}
  &\varepsilon \leq t \text{ for all }t\in T, \\
  &(n,m) \leq (n',m') \text{ if } m=m' \text{ and } n\leq n'.
\end{align*}
$\Tf_0$ is clearly well-founded, finite word $\omega$-automatic, and
satisfies $\infrank(\Tf_0)=1$ and $\rank(\Tf_0)=\omega$. 

Next, we show that for any $i,c\in\N$ and any given $\omega^i$-automatic
tree $\Tf$ there is also an $\omega^i$-automatic tree $\Tf'$ such that
$\infrank(\Tf')=\infrank(\Tf)\cdot c$ and $\rank(\Tf)=\rank(\Tf)\cdot c$. 

\begin{lemma}\label{lem:TreeRankFactorcViaConvolution}
  Let $c\in\N$ and $\Tf$ an  $\alpha$-automatic tree. 
  Then there is
  an $\alpha$-automatic tree $\Tf_c$ such that
  $\infrank(\Tf_c)=\infrank(\Tf)\cdot c$ and $\rank(\Tf)=\rank(\Tf)\cdot c$. 
\end{lemma}
\begin{proof}
  Let $\Tf = (T, \leq)$ and $L\subseteq T$ be the set of leaves of
  $\Tf$ ($L$ is $\alpha$-automatic because it is first-order definable
  if $\Tf$). 
  Set $T_c = \bigcup_{i=0}^{c-1} L^{\otimes i} \otimes T$ where
  $L^{\otimes 1} = L$ and $L^{\otimes {i+1}} = L^{\otimes i} \otimes
  L$. 
  The order of $\Tf_c$ is given by
  \begin{equation*}
    l_1 \otimes l_2 \otimes \dots \otimes l_i \otimes t \leq_c 
    l'_1 \otimes l'_2 \otimes \dots \otimes l'_j \otimes t'    
  \end{equation*}
  iff
  either $i=j, l_1=l'_1, \dots, l_i=l'_i$, and $t\leq t'$ or
  $i<j$ and $t\leq l'_{i+1}$. 
  
  Note that $\Tf_1=\Tf$ and $\Tf_{c+1}$ is obtained from $\Tf_c$ by
  attaching a copy of $\Tf$ to each leaf of $\Tf_c$. Thus, an easy
  induction on $c$ proves the claim. 
\qed \end{proof}

In the case $\alpha=\omega^i$ By replacing the convolution by
composition of $\omega^i$-words, we 
construct a finite word $\omega^{i+1}$-automatic
representation of the forest $\bigsqcup_{c\in\N} \Tf_c$ for any
$\omega^i$-automatic tree $\Tf$. 

\begin{lemma}
  For $\Tf$ a finite word $\omega^i$-automatic tree, the forest
  $\Ff:=\bigsqcup_{c\in\N} \Tf_c$ is finite word $\omega^{i+1}$-automatic.
\end{lemma}
\begin{proof}
  Let $\bot$ be a fresh symbol not occurring in the alphabet $\Sigma$ of 
  the representation of $\Tf$. 
  Let $W_c$ be the set of finite $\omega^{i+1}$-words whose letters
  all occur before position $\omega^{i}\cdot c$ and that have $\bot$
  exactly at position $\omega^{i}\cdot c$. We write $\bot_c$ for the
  word of $W_c$ whose only letter is $\bot$. 
  We identify an element of $\Tf_c$ with a word in $W_c$ as follows.
  Assume that $t_c\in \Tf_c$ has the form 
  $t_c= l_1\otimes \dots \otimes l_k \otimes t$ where each $l_i$ is a
  $\omega^i$-word denoting a
  leaf of $\Tf$ and $t$ is an $\omega^i$-word denoting an arbitrary
  element of $\Tf$.  
  Now let $t_c'$ be the word $l_1 + l_2 + \dots + l_k + t +
  \bot_{c-k-1}$ where $+$ denotes the concatenation of
  $\alpha$-words. Note that $t_c'\in W_c$.
  Since the order of two elements of $\Tf_c$ is defined by
  componentwise comparisons on the convolutions, this results in an
  $\omega^{i+1}$-automatic  presentation of $\Tf_c$ whose domain is a
  subset of $W_c$. 
  It is easy to see that the union of all these representations is
  an $\omega^{i+1}$-automatic forest. 
\qed \end{proof}

Of course, we can add a new  root to $\Ff$ and obtain an
$\omega^{i+1}$-automatic tree $\Tf'$ with 
$\infrank(\Tf')=\sup\{\infrank(\Tf)\cdot c \mid c\in \N\}$ and 
$\rank(\Tf')=\sup\{\rank(\Tf)\cdot c\mid c\in\N\}$.

Iterated application of this lemma to the tree $\Tf_1$ shows that for
each $i\in\N$ there is an $\omega^i$-automatic tree of rank
$\omega^{i+1}$ (and $\infrank$ $\omega^i$). Application of Lemma
\ref{lem:TreeRankFactorcViaConvolution} then proves the first part
of Theorem \ref{thm:TreesLowerBounds}.

We now use a variant of the previous construction in order to prove
the second part 
of Theorem \ref{thm:TreesLowerBounds}, i.e., we construct
$\alpha$-automatic trees of high ranks for ordinals $\alpha\geq
\omega^\omega$.  

\begin{definition}
  Let $\alpha$ be an ordinal. Let $D_\alpha$ be the set of finite
  $\alpha$-words $w$ over $\{\diamond,1\}$ such that for all limit
ordinals $\beta<\alpha$ 
  and all $c\in \omega$ the implication
  \begin{align*}
    w(\beta + c)=1 \Rightarrow w(\beta) = w(\beta+1) = 
    \dots = w(\beta+c) = 1
  \end{align*}
  holds. 
  We define a partial order on $D_\alpha$ via the suffix relation: for
  $w_1,w_2\in D_\alpha$ let
  $w_1 \suffix{\alpha} w_2$ if and only if
  for $\beta\leq \alpha$ maximal such that for all $0\leq \gamma <
  \beta \quad w_2(\gamma)=\diamond$ we have that
  $\forall \beta \leq \delta <\alpha \quad w_1(\delta)=w_2(\delta)$,
  i.e., the domain of $w_2$ is an upwards closed  subset of the domain
  of $w_1$ and  both agree on the domain of $w_2$. 
\end{definition}

Note that $\Tc_\alpha:=(D_\alpha, \suffix{\alpha})$ is
$\alpha$-automatic. 

\begin{lemma}
  $\Tc_\alpha:=(D_\alpha, \suffix{\alpha})$ is a tree. 
\end{lemma}
\begin{proof}
  Since $D_\alpha$ contains finite $\alpha$-words $w$ there are only
  finitely many positions $\beta<\gamma$ with $w(\beta)=1$. Thus,
  there are also only finitely many 
  suffixes of $w$ that are undefined up to some position in $\supp(w)$. 
  This implies that all
  ascending chains are finite. 
  Moreover, the suffix relation is a linear order when restricted to
  the suffixes of a fixed word $w$. 
\qed \end{proof}

The following lemma combined with 
Lemma \ref{lem:TreeRankFactorcViaConvolution} proves the second part
of Theorem \ref{thm:TreesLowerBounds}. 

\begin{lemma}
  For all ordinals $\alpha, \alpha'$ such that $\alpha=\omega\cdot
  \alpha' \geq \omega$,
  $\rank(\Tc_\alpha) =\alpha'$. 
\end{lemma}
\begin{proof}
  The proof is by induction on $\alpha'$. 
  For $\alpha=\omega \cdot 1 = \omega$ note that $D_\alpha$ consists
  of all words $1^m\diamond^\omega$, $m\in\N$ where the word
  $\diamond^\omega$ is suffix 
  of all other elements. Moreover, these others are pairwise
  incomparable. Thus, $\Tc_\omega$ is the infinite tree of depth $1$
  which has rank $1$ as desired. 
  We now proceed by induction. 
  \begin{enumerate}
  \item Assume that $\alpha'$ is a successor ordinal, i.e., there is
    some $\beta'$ such that $\alpha=\omega\cdot \alpha' = \omega \cdot
    \beta' + \omega$.
    Note that the words directly below $\diamond^\alpha$ are those of the
    form $w=\diamond^\gamma1^m\diamond^\delta$ such that
    $\gamma+\delta=\alpha$ and 
    $\gamma$ is some limit ordinal and $m < \omega$. 
    Fix such a word and note that 
    $D_\alpha\cap \{w'\mid w'\suffix{\alpha} w\}$ induces a suborder
    isomorphic to $(D_\gamma, \suffix{\gamma})$ which by induction
    hypothesis has rank $\gamma'$ for $\gamma'$ such that
    $\gamma=\omega\cdot \gamma'$. 
    Thus, the suborders of maximal rank $\beta'$ are induced by the
    elements $w_m= \diamond^{\omega\cdot \beta'}1^m\diamond^\omega$ for each
    $m < \omega$. Since these are infinitely many nodes of $\infrank$
    $\beta'$, the rank of $\Tc_\alpha$ is $\beta'+1=\alpha'$.
  \item 
    Assume that $\alpha'$ is a limit ordinal and
    $(\beta_i)_{i\in\omega}$ converges to $\alpha'$ and
    $\beta_i<\alpha$ for each $i\in\omega$. 
    Then each $w^m_i:=\diamond^{\beta_i}1^m\diamond^\alpha$ for $m,i\in\omega$ is
    directly below $\diamond^\alpha$ 
    and induces a suborder isomorphic to
    $(D_{\beta_i}, \suffix{\beta_i})$ of $\infrank$ $\beta_i$. 
    Thus, $\infrank(\Tc_\alpha)\geq \alpha'$. 
    But as in the previous case we see that all proper suborders have
    $\infrank<\alpha$ whence $\infrank(\Tc_\alpha)\leq \alpha'$. 
    Thus, its $\infrank$ is exactly $\alpha'$. 
  \end{enumerate}
\qed \end{proof}

\section{ Finite-Rank-Scattered-Automatic Ordinals }
\label{sec:OrdinalsOverFiniteRankOrders}

In this section we prove that for every scattered linear $\Lf$
order such that $\FC(\Lf)=\FCstar(\Lf)= n<\omega$, the
$\Lf$-oracle-automatic ordinals 
are exactly those below $\omega^{\omega^{n+1}}$. 
For this purpose, it suffices to show that
$\omega^{\omega^{n}}$ is $\Lf$-oracle-automatic. Since
$\Lf$-oracle-automatic structures are closed under finite
(lexicographically ordered) products, it
follows that for each $k$ the ordinal
$(\omega^{\omega^{n}})^k = \omega^{\omega^{n} \cdot k}$ is
$\Lf$-oracle-automatic. Since the $\Lf$-oracle-automatic ordinals are
closed under definable substructures we conclude that all ordinals below
$\omega^{\omega^{n+1}}$ are $\Lf$-oracle-automatic.

\begin{theorem}
  Let $\Lf$ be a scattered linear order with
  $\FC(\Lf)=\FCstar(\Lf)=1+n<\omega$.  
  The ordinal $\omega^{\omega^{n}}$ is finite word $\Lf$-oracle-automatic.
\end{theorem}
\begin{proof}
  We inductively prove the following claim:
  For each $n$ there is are finite automata $\Ac_n$ and $\Bc_n$ 
  such that for
  every scattered linear order $\Lf$ with
  $\FC(\Lf)=\FCstar(\Lf)=1+n$ there is an $\Lf$-oracle $o_\Lf$ 
  such that $\omega^{\omega^{n}}$ is $\Lf$-$o_\Lf$-automatic 
  where $\Ac_n$ recognises the domain and $\Bc_n$ the order $<$ in
  this representation. Moreover the empty $\Lf$-word represents $0$).
  
  In the base $n=0$, we distinguish two cases:
  \begin{enumerate}
  \item $\Lf=\omega$ or $\Lf=\Z$: Let $o_\Lf:\Lf\to\{\diamond,1\}$ be
    an oracle such that $\supp(o_\Lf)$ is isomorphic to $\omega$. 
    Let $\Ac_0$ accept all $\Lf$-words $w$ such that
    $\supp(w)\subseteq \supp(o_\Lf)$ and $\lvert \supp(w)\rvert = 1$. 
    The order is given by $w<v$ iff $\supp(w)$ is to the left of
    $\supp(v)$.
  \item $\Lf=\omega^*$: Let $o_\Lf:\omega^*\to \{\diamond,1\}$ be the
    constant $1$ oracle. 
    Again, the domain recognised by $\Ac_0$ consists of all
    $\Lf$-words $w$ such that 
    $\supp(w)\subseteq \supp(o_\Lf)$ and $\lvert \supp(w)\rvert = 1$. 
    The order is given by $w<v$ iff $\supp(w)$ is to the right of
    $\supp(v)$.
  \end{enumerate}
  There is an  automaton $\Bc_1$ recognising the order independent of
  the shape of $\Lf$. If $\Bc_1$ applies a right limit transition it
  guesses whether $\supp(o_\Lf)$ is defined arbitrarily close to the
  minimal cut. This guess can be checked at the successor
  transitions. Depending on its guess, it recognises the correct order
  according to the case distinction. Since both orders are automatic,
  this combined order is also automatic. 

  For the induction step assume that the claim was proved for all $n'<n$. 
  Let $\Lf$ be some scattered linear order with
  $\FC(\Lf)=\FCstar(\Lf)=1+n$. 
  Recall the automaton $\Cc_{n+1}$ from definition
  \ref{def:AutLimitStages} which determines the left and right order
  of each cut of $\Lf$. 
  Using those cuts where $\Cc_n$ is in a state from 
  $M_{n}=\{n\}\times \{0,1, \dots, n\} \cup 
  \{0, 1, \dots, n\} \times \{n\}$, we obtain a 
  decomposition 
  $\Lf=\sum_{i\in\Z} \Lf_i$ such that $\FC(\Lf_i)=n$ and
  there is an infinite ascending (or descending) sequence
  \begin{align}\label{align:seqOrdinalOrders}
    i_0 < i_1 < i_2 < \dots\text{ such that }\FC(\Lf_i)=n \text{
      and } \forall j\geq 1\quad i_j-{i_{j-1}}\geq 2.
  \end{align}  
  By this we mean that $\Cc_n$ upon reading any $\Lf$-word is not in a
  state from $M_{n}$ on any cut strictly in $\Lf_i$ but it is in one of the
  states from $M_{n}$ at the last cut before and the first cut after
  $\Lf_i$. 
  
  We now describe the case of an ascending chain, but the descending case
  is analogous. 
  Let $o_\Lf$ be the oracle defined by 
  $o_\Lf(x)=(1,o_{\Lf_{i_j}}(x))$ if $x\in \Lf_{i_j}$ 
  and $o(x)=\diamond$ if for all $j\in\N$ we have 
  $x\notin \Lf_{i_j}$.
  The domain of our presentation of $\omega^{\omega^n}$ consists of
  those finite $\Lf$-words $w$ such that 
  $\supp(w)\subseteq \bigcup_{j\in\N} \Lf_{i_j}$ and for each $j$ 
  $\Ac_{n-1}$ accepts $w$ restricted to $\Lf_{i_j}$. 
  This set is recognised by an $\Lf$-$o_\Lf$-automaton $\Ac_n$ as
  follows. $\Ac_n$ simulates $\Cc_n$. At the initial state and
  whenever $\Cc_n$ is in a state from $M_{n}$, it guesses whether
  the next part of $\Lf$ is one of the $\Lf_{i_j}$ where $o_\Lf$ is
  defined. In this case, it starts a simulation of $\Ac_{n-1}$. This
  simulation is stopped when $\Cc_n$ is again in a state from
  $M_{n}$. If it starts a simulation of $\Ac_{n-1}$ and $o_\Lf$
  turns out to be undefined on this part, then the run is
  aborted. Analogously, the run is aborted if we did not start a
  simulation of $\Ac_{n-1}$ and reach a position in $\supp(o_{\Lf})$. 
  
  We identify each word $w$ accepted by $\Ac$ with a sequence
  $(\alpha_j)_{j\in \omega}$ of ordinals in $\omega^{\omega^{n-1}}$ such that
  all but finitely many $\alpha_j$ are $0$ and 
  $\alpha_j$ is the ordinal represented by the restriction of $w$ to
  the $\Lf_{i_j}$ (with respect to the order induced by the order
  automaton $\Bc_{n-1}$). Of course there is an automaton $\Bc'_n$ that
  orders the 
  sequences $(\alpha_j)_{j\in \omega}$ backwards lexicographically, i.e., 
  $(\alpha_j)_{j\in \omega} < (\beta_j)_{j\in\omega}$ if and only if there is
  $j_0\in \N$ such that $\alpha_j=\beta_j$ for all $j> j_0$ and
  $\alpha_{j_0} < \beta_{j_0}$. This order is $\Lf$-$o_\Lf$-automatic
  (just apply order automaton $\Bc_{n-1}$ on the part corresponding to
  $\Lf_{i_j}$ indicated by the oracle and remember the last outcome
  different from '{$=$}'). 
  This gives a presentation of 
  $\sum_{i\in\omega} (\omega^{\omega^{n-1}})^i = 
  \sum_{i\in\omega} (\omega^{\omega^{n-1} \cdot i})  =
  \omega^{\omega^n}$.
  Note that the definition of the oracle $o_\Lf$ depends on $\Lf$ but
  the automata do not depend on $\Lf$. We only need a a slightly different
  order in the case of a descending sequence instead of the
  ascending sequence in  \eqref{align:seqOrdinalOrders}. In the case
  of a descending sequence the order automaton uses lexicographic
  ordering instead of backwards lexicographic order because the domain
  of   the presentation can be identified with
  $(\alpha_j)_{j\in\omega^*}$. 
  Of course, we can define an automaton $\Bc_n$ that guesses (and
  verifies) whether
  $\supp(o_\Lf)$ is cofinal and depending on this guess, simulates
  $\Bc'_n$ or the variant $\Bc''_n$ performing lexicographic
  ordering (since $o_\Lf$ is either coinitial or cofinal, 
  the correctness of this guess can be checked immediately in the
  transitions leaving the initial states).
\qed \end{proof}

\begin{problem}
  Can one lift the previous theorem to orders of transfinite rank?
\end{problem}

\section{The countable atomless Boolean algebra is not ordinal automatic}\label{sec: Boolean algebra}

If $\eta$ is an ordinal, there is an apparent bijection between
$\Cuts(\eta)$ and the ordinal $\eta+1=\{\alpha\mid \alpha\leq
\eta\}$ which we will use to identify cuts. 
Let $\Cuts^-(\eta)=\Cuts(\eta)\setminus \{(\eta,\emptyset)\}$. 
We call $w:\eta \to \{\diamond\}$ the \emph{empty input}. 
If $r\colon \Cuts^-(\eta)\rightarrow \mathcal{S}$ is a run of $\mathcal{A}$
with $\gamma\leq\eta$ and $\gamma<\eta$ is a limit ordinal, let as above
$\lim_{\gamma^-} r$ 
denote the set of states appearing unboundedly
often before $\gamma$. 

\begin{lemma}\label{lemma: pumping for ordinal automata} (Pumping)
  Let
  $\mathcal{A}$ be a non-deterministic ordinal automaton with state
  set $S$, $m\in\N$, and $\gamma$ some ordinal
  with $\gamma\neq 0$. 
  Suppose $S_{\lim}\subseteq S^+\subseteq S$ and $s\in S^+$ with 
  $|S_{lim}|\leq m$. 

  If there is a run 
  \begin{align*}
    r\colon \Cuts(\omega^{m})\rightarrow S^+    
  \end{align*}
  on empty input with
  $r(0)=s$, 
  $\lim_{\omega^{m-}} r= S_{\mathsf{lim}}$, and
  $r(\Cuts(\omega^{m}))=S^+$, 
  then 
  there is a run 
  \begin{align*}
    \bar r\colon \Cuts(\omega^{m}\gamma)\rightarrow
    S^+    
  \end{align*}
  on empty input with
  $\bar r(0)=s$, 
  $\lim_{\omega^{m}\gamma^-} \bar r= S_{\mathsf{lim}}$, and
  $\bar r(\Cuts(\omega^{m}\gamma))=S^+$. 
\end{lemma} 

\begin{proof} 
  The proof is by induction on $m$ and $\lvert S_{lim}\rvert$. 
  \begin{itemize}
  \item First 
    suppose that $S_{lim}=S^+$. Then
    $r(\omega^{m})=r(\alpha_0)$ for 
    some $\alpha_0<\omega^{m}$. Let
    $\bar{r}(\alpha)=r(\alpha)$ and 
    $\bar{r}(\omega^m\beta+\alpha)=r(\alpha_0+\alpha)$ for
    $\alpha<\omega^m$ and $1\leq \beta<\gamma$. Let
    $\bar{r}(\omega^m\gamma)=r(\omega^m)$. 
  \item 
    Now suppose that 
    $S_{lim}\subsetneq S^+$. 
    Choose
    $n_0$ with
    $r([\omega^{m-1}n_0,
    \omega^{m}))\subseteq S_{lim}$.   
    \begin{itemize} 
    \item If there is $n\geq n_0$ with
      $r([\omega^{m-1}n,\omega^{m-1}(n+1)))=S_{lim}$,  
    choose $\beta_0\in [\omega^{m-1}n,\omega^{m-1}(n+1)$ with
    $r(\beta_0)=r(\omega^{m-1}(n+1))$.  
    Let $\bar{r}(\alpha)=r(\alpha)$ for $\alpha\leq\omega^{m-1}(n+1)$, 
    let $\bar{r}(\omega^{m-1}\beta+\alpha)=r(\beta_0+\alpha)$ for
    $\alpha<\omega^{m-1}$ and $\omega \beta<\gamma$, and let
    $\bar{r}(\omega^m\gamma)=r(\omega^m)$.  
    \item If there is no such $n$, find $n_0=\beta_0<\beta_1<...$ with
      $\sup_{i\in\omega}\omega^{m-1}\beta_i=\omega^m\gamma$.  
    Let $\bar{r}(\alpha)=r(\alpha)$ for $\alpha\leq\omega^{m-1} n_0$. 
    We can pump $r\upharpoonright [\omega^{m-1}
    n,\omega^{m-1}(n+1)]$ to a run  
    $\bar{r}\colon [\omega^{m-1}\beta_n, \omega^{m-1}\beta_{n+1}]
    \rightarrow S^+$ for $n\geq n_0$ by the induction hypothesis for
    smaller $S_{lim}$.   
    \end{itemize} 
    \end{itemize}
\qed \end{proof}

\begin{lemma}\label{lemma: Shrinking for ordinal automata} (Shrinking)
  Let
  $\mathcal{A}$ be a non-deterministic ordinal automaton with state
  set $S$, $m\in\N$, and $\gamma$ some ordinal
  with $\gamma\neq 0$. 
  Suppose $S_{\lim}\subseteq S^-\subseteq S$ and $s\in S^-$ with
  $|S^-|\leq m$.  
  If there is a run 
  \begin{align*}
    r\colon \Cuts^-(\omega^{m}\gamma)\rightarrow  S^-    
  \end{align*}
  on empty input with
  $r(0)=s$, 
  $\lim_{\omega^{m}\gamma^-} r= S_{\mathsf{lim}}$, and 
  $r(\Cuts^-(\omega^{m}\gamma))=S^-$, 
  then 
  there is a run 
  \begin{align*}
    \bar r \colon \Cuts^-(\omega^{m})\rightarrow S^-     
  \end{align*}
  on empty input with
  $\bar r(0)=s$, 
  $\lim_{(\omega^{m})^-} \bar r= S_{\mathsf{lim}}$, and 
  $\bar r(\Cuts^-(\omega^{m}))=S^-$. 
\end{lemma} 

\begin{proof} 
  The proof is by induction on $m$, $\gamma$, and the size of $S^-$. 
  The claim is obvious for $m= 1$ or $\gamma=1$. Thus, we assume that
  $\gamma\geq 2$ and $m \geq 2$. 
  
  \begin{itemize}
  \item First suppose that $S_{\mathsf{lim}}=S^-$  and that there
    is some $\beta<\gamma$ with $\lim_{(\omega^m\beta)^-} r = S^-$. 
    Then we can shrink the run $r\upharpoonright \omega^m\beta$ to a
    run $\bar{r}\colon \omega^m\rightarrow S^-$ by the induction
    hypothesis for $\beta$.  
  \item 
    Next suppose that $S_{\mathsf{lim}}=S^-$ and that for
      each $\beta<\gamma$, there is an $s\in S^-$ such that 
      $s\notin \lim_{(\omega^m\beta)^-} r$.
      There are the following subcases:
      \begin{itemize}
      \item First suppose that $\gamma=\bar{\gamma}+1$. Choose $\beta_0\in [\omega^m\bar{\gamma}, \omega^m \gamma))$ with $r(\beta_0)=r(0)$. 
      Let $\bar{r}(\alpha)=r(\beta_0+\alpha)$ for $\alpha<\omega^m$. 
      \item Suppose that  $\gamma=\omega$.
        By assumption, for each $i$, there is some $\alpha_i$ with
        $\omega^m i \leq \alpha_i < \omega^m (i+1)$ and a state
        $s_i\in S^-$ such  that  $s_i\neq r(\beta)$ for all
        $\alpha_i \leq \beta < \omega^m(i+1)$. Thus, we can apply the
        induction  hypothesis for smaller $S^-$ to each
        $r{\restriction}[\alpha_i, \omega^m(i+1))$ 
        and shrink it to a run of size $\omega^{m-1}$. 
        Note that the length of $r{\restriction}[\omega^m i,
        \alpha_i]$ is also bounded by some $\omega^{m-1}\cdot k_i$. 
        Thus, composition of these runs yields the desired run of
        length $\omega^m$. 
      \item Finally, suppose that $\gamma>\omega$ is a limit ordinal
        and that 
        $\gamma_1<\gamma_2<\dots <\gamma$ are ordinals such that 
        $\lim \gamma_i=\gamma$.   
        By induction hypothesis for smaller $\gamma$, we can 
        shrink each run 
        $r{\restriction}[\omega^m\gamma_i, \omega^m\gamma_{i+1})$ to a
        run of length $\omega^m$ such that 
        each state of $S^-$ appears in one of these runs. 
        Composition of the resulting runs reduces this case to the
        previous case. 
      \end{itemize}
      
    \item Finally,
      suppose that $S_{\mathsf{lim}}\subsetneq S^-$. Let 
      $\alpha_0$ denote the least $\alpha<\omega^m\gamma$ such that
      only states $s\in S_{\mathsf{lim}}$ appear in
      $[\alpha,\omega^m\gamma)$. 
      There are two subcases:
      \begin{itemize}
      \item First suppose that $\alpha_0<\omega^m\beta$ for some
        $\beta<\gamma$. 
        Note that $[\alpha_0,\omega^m\beta)$ and
        $[\alpha_0,\omega^m\gamma)$ are of the form 
        $\omega^{m}\cdot \delta$ with $\delta\leq\gamma$. 
        Since the image of 
        $r\upharpoonright [\alpha_0,\omega^m\gamma)$ is contained in
        $S_{\mathsf{lim}}\subsetneq S^-$, 
        we can shrink 
        $r\upharpoonright [\alpha_0,\omega^m\gamma)$ to a run
        $\bar{r}\upharpoonright 
        [\alpha_0,\omega^{m}\beta)$ by the induction
        hypothesis for smaller $S^-$. 
        Since $\beta<\gamma$ we conclude by  application of the
        induction hypothesis to this shorter run. 
      \item 
        Second suppose that  $\alpha_0 \geq \omega^m\beta$ for all
        $\beta<\gamma$. We conclude immediately that 
        $\gamma$ is a successor, i.e., $\gamma=\bar{\gamma}+1$ and 
        $\alpha_0\geq  \omega^m \bar{\gamma}$. 
        Now we distinguish the following cases.
        \begin{enumerate}
        \item Assume that $\bar \gamma=1$ and  that 
          $r(\omega^m) \in \lim_{(\omega^m)^-} r$. Then there is a
          $\beta <\omega^m$ such that $r(\beta)=r(\omega^m)$ and for
          each state $s\in S^-$ such that $s$ occurs in $r$ strictly before
          $\omega^m$  also occurs before $\beta$. 
          Then the composition of $r{\restriction}[0, \beta)$ with
          $r{\restriction}[\omega^m, \omega^m\gamma]$ yields the
          desired run. 
        \item Assume that $\bar\gamma=1$ and that
          $r(\omega^m) \notin \lim_{(\omega^m)^-} r$. Thus, there is
          some $\beta<\omega^m$ such that 
          $r([\beta, \omega^m))\subseteq 
          S^- \setminus \{r(\omega^m)\}$.
          Thus, we can apply the induction hypothesis for smaller $m$
          and $S^-$ shrinking $r{\restriction}[\beta, \omega^m)$ 
          to a run $\hat r$ on domain $[\beta, \beta+\omega^{m-1})$ with 
          $\hat r([\beta, \beta+\omega^{m-1})) = r([\beta, \omega^m))$
          and
          $\lim_{(\beta+\omega^{m-1})^-} \hat r = \lim_{(\omega^m)^-}
          r$.
          Since $\beta<\omega^{m-1}\cdot k$ for some $k\in\N$, 
          Composition of $r{\restriction}[0, \beta)$ with $\hat r$ and
          $r{\restriction}[\omega^m, \omega^m\gamma)$ yields the
          desired run $\bar r$ of length $\omega^m$.
        \item 
          If $\bar\gamma>1$, we  apply the induction hypothesis (for
          smaller $\gamma$) to 
          $r{\restriction}[0, \omega^m\bar \gamma)$ and shrink this
          run to a run $\bar r$ of length $\omega^m$. The composition of
          $\bar r$ and  $r{\restriction}[\omega^m\bar \gamma,
          \omega^m\gamma)$ is a run of length $\omega^m\cdot 2$ and we
          can apply the induction hypothesis for smaller $\gamma$.
        \end{enumerate} 
      \end{itemize}
    \end{itemize}
\qed \end{proof}

We directly obtain the following corollary.
\begin{corollary} 
  Let $\gamma\geq 1$ be an ordinal and let
  $\Ac_\gamma=\Ac_1=(S, \Sigma, I,   F, \Delta)$ be an
  automaton (where we interpret $\Ac_i$ as  $\omega^mi$-automaton). 
  For all $s_0,s_1\in S$, 
  \begin{align*}
    s_0 \run{\diamond^{\omega^m}}{\Ac_1} s_1 \Longleftrightarrow 
    s_0 \run{\diamond^{\omega^m\gamma}}{\Ac_\gamma} s_1       
  \end{align*}
  where $\diamond^\alpha$ denotes the empty input of length $\alpha$. 
\end{corollary}

A formula is $\Sigma_0$ if it is quantifier-free. A formula is $\Pi_i$
if it is logically equivalent to the negation of a $\Sigma_i$-formula.
Formulas of the form $\exists x_0...\exists
x_n\varphi(x_0,...,x_n,y_0,...,y_k)$ for some $\Pi_i$-formula
$\varphi$ are $\Sigma_{i+1}$.

\begin{lemma}\label{lemma: elementary substructure} 
Let $L_0, L_1, \dots, L_k$ be linear orders and let
$\delta_1, \delta_2, \dots, \delta_k, \eta_1, \eta_2, \dots, \eta_k$
be ordinals all strictly greater than $0$. 
Let $\Ac, \Ac_{R_1}, \dots, \Ac_{R_n}$ be
automata such that $m\in\N$ is a bound on the number of states of any
of these. Let $n_0\in\N$ be some number. 
Setting $K^j_i:=\omega^{m+n_0}\cdot j_i$ for $j\in\{\delta,\eta\}$
define the maps 
\begin{align*}
  f_\delta: \prod_{i=0}^k L_i &\to \delta:= 
  L_0+\sum_{i=1}^k (K^\delta_i+ L_i),\\
  (w_1, \dots, w_k) &\mapsto w_1 + \diamond^{\omega^{m+i}\cdot \delta_1}
  + w_2 + \dots + \diamond^{\omega^{m+i}\cdot \delta_k} + w_k, \text{ and}\\
  f_\eta: \prod_{i=0}^k L_i &\to \eta:=
  L_0+\sum_{i=1}^k (K^\eta_i+ L_i), \\
  (w_1, \dots, w_k) &\mapsto w_1 + \diamond^{\omega^{m+i}\cdot \eta_1}
  + w_2 + \dots + \diamond^{\omega^{m+i}\cdot \eta_k} + w_k.  
\end{align*}
Let $M_i$ be the finite word $i$-automatic structure induced by 
$\Ac, \Ac_1, \dots, \Ac_n$  for $i\in\{\delta, \eta\}$. 
For every $\Sigma_i\cup\Pi_{n_0}$-formula $\varphi(\vec{x})$ and all
$\vec{w}=(\vec{w_0}, \dots, \vec{w_k})\in (\prod_{i=0}^k
W_i)^{<\omega}$ (where $W_i$ denotes the set of finite $L_i$-words), 
$M_\delta\models \varphi(f_\delta(\vec{w}))$ if and only if
$M_\eta \models \varphi(f_\eta(\vec{w}))$. 
\end{lemma}

\begin{proof} 
  The claim for $n_0=0$ follows from the Pumping and Shrinking
  Lemmas because we can translate any run on $\diamond^{\omega^m\cdot
    \gamma}$ into a run on $\diamond^{\omega^m\cdot \gamma'}$ with
  same initial and final state for all ordinals $\gamma, \gamma'\geq
  1$. 
  
  For the inductive step, assume that the claim holds for all
  $n'<n_0\in\N$. 
  
  Due to symmetry of the claim and since every $\Pi_{n_0}$-formula is
  the negation of a $\Sigma_{n_0}$-formula it suffices to prove that
  $M_\eta \models \varphi(f_\eta(\vec{w}))$ if
  $M_\delta\models \varphi(f_\delta(\vec{w}))$ for a $\Sigma_{n_0}$
  formula $\varphi$.

  Let  $\varphi$ be some $\Pi_{{n_0}-1}$-formula and 
  $\vec{v}=f^\delta(\vec{w})$ for some $\vec{w}_i
  \in  (\prod_{i=0}^k L_i)^{<\omega}$ such that 
  $N\vDash \exists \vec{x} \ 
  \varphi(\vec{x}, \vec{v})$. Choose  
  $\vec{t}\in {M_\delta}^{<\omega}$ with $M_\delta \vDash \varphi(t,\vec{v})$.
  
  Since $\vec{t}$ has finite support, for each $1\leq i \leq k$, 
  $\supp(\vec{t})\cap K^\delta_i$ induces a decomposition
  $K^\delta_i= L^i_0+ \sum_{j=1}^m (\bar K^\delta_j+ L^i_j)$ such that
  \begin{itemize}
  \item $\bar K^\delta_j = \omega^{n_0-1}\cdot \kappa$ for some ordinal
    $\kappa\geq 1$,
  \item $L^i_j = \omega^{n_0-1}$, and
  \item   $\supp(\vec{t})\cap  K^\delta_i \subseteq \bigcup_{j=1}^m L^i_j$, 
  \end{itemize}
  Fix ordinals $\bar K^\eta_j$ such that
  $K^\eta_i= L^i_0+ \sum_{j=1}^m (\bar K^\eta_j+ L^i_j)$ (these exist
  because $K^\eta_i = \omega^{m+{n_0}-1} \cdot (\omega\cdot \kappa)$
  for some ordinal $\kappa\geq 1$).
  
  Application of the inductive hypothesis to
  $\varphi$ and the functions
  \begin{align*}
  &g^\delta: L_0\times \prod_{i=1}^k (\prod_{j=1}^m L^i_j \times L_i) \to \delta,
  \text{ and }\\
  &g^\eta: L_0\times \prod_{i=1}^k (\prod_{j=1}^m L^i_j \times L_i) \to \eta
  \end{align*}
  defined in the apparent way
  shows
  that 
  there are words $\vec{w_1}, \vec{w_2}$ such that
  $g^\delta(\vec{w_1})= \vec{v}, 
  g^\delta(\vec{w_2})= \vec{t}$, and
  \begin{align*}
    M_\delta\models \varphi( g^\delta(\vec{w_1}), g^\delta(\vec{w_2}))
    \text{  if and only if }
    M_\eta\models \varphi( g^\eta(\vec{w_1}), g^\eta(\vec{w_2})).
  \end{align*}
  By definition, one easily sees that
  $g^\eta(\vec{w_1}) = f^\eta(\vec{w}) =
  \vec{v}$. Thus, 
  \begin{align*}
    M_\eta \models \exists \vec{x} \ 
    \varphi(\vec{x}, f^\eta(\vec{w})).
  \end{align*}
\qed 
\end{proof}


\begin{definition} 
  For any ordinal $\alpha$, let $\bar{\alpha}$ be the  ordinal of the
  form
  $\bar{\alpha}=\omega^{m+1}\beta$ for some ordinal $\beta$ such
  that $\alpha=\bar{\alpha}+\omega^m n_m+\omega^{m-1} n_{m-1}+...+n_0$ and
\begin{enumerate} 
\item 
\begin{enumerate}[a.] 
\item Let $U_m(\alpha)$ denote the set of ordinals $\gamma=\bar{\alpha}+\omega^m l_m+\omega^{m-1} l_{m-1}+...+l_0$ such that either 
\begin{itemize}
\item $\gamma=\alpha$ or 
\item $l_k\leq n_k+m$ and  $l_i\leq m$ for all $i<k$, 
\end{itemize} 
where $k$ is maximal with $l_k\neq n_k$. 
\item Let $U_m(X)=\bigcup_{\gamma\in X\cup\{0\}}U_m(\gamma)$. 
\item Let $U_m(X,\delta)=U_m(X\cup\{\delta\})\cap \delta$.  
\end{enumerate} 
\item 
\begin{enumerate}[a.] 
\item Let $c_m(\alpha)=\max_{i\leq m} n_i$. 
\item Let $c_m(X)=\max_{\gamma\in X}c_m(\gamma)$. 
\end{enumerate} 
\item Let $d_m(X)=|\{\bar{\gamma}\mid \gamma\in X\cup\{0\}\}|$. 
\end{enumerate} 
\end{definition} 

Let $U_m^1(X)=U_m(X)$ and $U_m^{i+1}(X)=U_m(U_m^i(X))$ for $i\in\N$, 
and similarly let $U_m^1(X,\delta)=U_m(X,\delta)$ and $U_m^{i+1}(X,\delta)=U_m(U_m^i(X,\delta),\delta)$ for $i\in\N$. 
A rough upper bound for the sizes of these sets is given in the following lemma.

\begin{lemma} Suppose that $X$ is a finite set of ordinals and $i\geq
  1$. Then
  \begin{align*}
    &|U_m^i(X)|\leq (c_m(X)+im)^{m+1} d_m(X), \text{ and also}\\
    &\lvert U_m^i(X,\delta)|\leq (c_m(X\cup\{\delta\})+im)^{m+1}
    d_m(X\cup\{\delta\}).    
  \end{align*}
\end{lemma} 

\begin{proof} The coefficient of $\omega^j$ of an element of $U_m^i(\gamma)$ can take at most $(c_m(w)+im)^{m+1}$ many different values for any fixed $j\leq m$. 
Hence $|U_m^i(\alpha)|\leq (c_m(w)+im)^{m+1}$ for all ordinals $\alpha$ and all $i\geq 1$. Moreover $d_m(U_m^i(X))=d_m(X)$ for all $i\geq 1$. 
\qed \end{proof} 

It follows that there are at most
$|\Sigma_{\diamond}|^{(c_m(X)+im)^{m+1} d_m(X)}$ many finite words $w$
over alphabet $\Sigma_\diamond$ with $\mathsf{supp}(w)\subseteq U_m^i(X)$ for $i\geq 1$, where $\Sigma_{\diamond}$ is an alphabet with $\diamond\in\Sigma$. 

A relation $R\subseteq X\times Y$ is called \emph{locally finite} if for every $x\in X$, there are at most finitely many $y\in Y$ with $(x,y)\in R$. 

\begin{lemma}(Growth lemma) Suppose $\eta$ is an ordinal and $R\subseteq (\Sigma^*)^k\times(\Sigma^*)^l$ is a locally finite relation of finite $\eta$-words. Suppose $R$ is recognised by an $\eta$-automaton $\mathcal{A}$ with at most $m$ states. Then $\mathsf{supp}(w)\subseteq U_{m+1}(\mathsf{supp}(v),\eta)$ for all $(v,w)\in R$. 
\end{lemma} 

\begin{proof} Suppose $\alpha\in \mathsf{supp}(w)\setminus U_{m+1}(\mathsf{supp}(v),\eta)$ is minimal. Let $k\in \N$ be least such that there are $\beta\in \mathsf{supp}(w)\cup\{0,\eta\}$ and $\delta$ with $\omega^{k+1}\delta\leq \alpha,\beta<\omega^{k+1}(\delta+1)$. It follows from the Pumping Lemma that $k\leq m$. Choose the maximal such $\beta$ for this $k$. If $\beta\neq\delta$ then $\mathsf{supp}(w)\cap (\beta,\omega^{k+1}(\delta+1))=\emptyset$. At most $m$ different states can appear at the ordinals $\beta+\omega^{k}l$ for $l\in\N$. Since $R$ is locally finite, none of the states appears twice between $\beta$ and $\alpha$ if $\beta\neq\delta$, since otherwise it is possible to shrink the run, and hence $\omega^{k+1}\delta \leq \alpha<\beta+\omega^k m$. 
If $\beta=\delta$ then $\omega^{k+1}\delta\leq\alpha<\beta$. 
Let $\alpha_j$ denote the coefficient of $\omega^j$ in the Cantor normal form of $\alpha$ for $j\in\N$. Since there are at most $m$ states and $R$ is locally finite, $\alpha_j<m$ for all $j<k$ and hence $\alpha\in U_{m+1}(\beta)$. 
\qed \end{proof} 

Let $\lceil x \rceil$ denote the least $n\in\omega$ with $x\leq n$, and $\log$ the logarithm with base $2$. 

\begin{lemma} (Growth lemma for monoids) Suppose the multiplication of the monoid $(M,\cdot)$ is recognised by an automaton with $\leq m$ states. Suppose $s_1,...,s_{n}\in M$ and $\mathsf{supp}(s_i)\subseteq X$ for $1\leq i\leq n$ where $n\geq 2$. Then $\mathsf{supp}(s_1\cdot...\cdot s_n)\subseteq U_{m+1}^{\lceil \log n\rceil}(X)$. 
\end{lemma} 

\begin{proof} We follow the proof of \cite[Lemma
  3.2]{DBLP:journals/lmcs/KhoussainovNRS07}. The statement follows
  from the Growth Lemma for $n\geq 2$. For $n>2$ let $k=\lceil
  \frac{n}{2}\rceil$ and $l=n-k$. Then $\lceil \log k\rceil, \lceil
  \log l\rceil< \lceil \log n\rceil$. Let $t=s_1\cdot...\cdot s_k$ and
  $u=s_{k+1}\cdot...\cdot s_n$. Then $\mathsf{supp}(t)\cup
  \mathsf{supp}(u)\subseteq U_{m+1}^{\lceil \log n\rceil-1}(X)$ by the
  induction hypothesis for $\lceil\frac{n}{2}\rceil$. Thus, 
  $\supp(t\cdot u) \subseteq U_{m+1}^{\lceil
    \log n\rceil}(X)$ by the Growth Lemma applied to $t$ and $u$.  
\qed \end{proof}

We  prove that the countable atomless Boolean algebra is not
$\delta$ automatic for any ordinal $\delta$. 
We first conclude by Lemma
\ref{lemma: elementary substructure} that it suffices to consider
ordinals of the form 
$\delta=\omega^k$ with $k\in\N$.

\begin{corollary}
  Let $\eta, \kappa, \delta$ be ordinals such that
  $\eta\geq 1$, $\kappa<\omega^\omega$, and
  $\delta=\omega^{\omega}\cdot \eta + \kappa$. If the countable
  Boolean algebra is finite word $\delta$-automatic then 
  it is finite word $\omega^k$-automatic for some $k\in\N$. 
\end{corollary}
\begin{proof}
  Let $\bar\Ac=(\Ac, \Ac_\cup, \Ac_\cap, \Ac_0, \Ac_1)$ be $\delta$-automata
  representing  the countable atomless Boolean algebra 
  $\mathfrak{M}=(M, \cup,\cap, \bf{0}, \bf{1})$.
  Let $n\in\N$ be a bound on the number of states of any of these automata.
  $\delta$ can be written as a sum
  $\omega^{n+2} + \omega^\omega \cdot \eta + \kappa$. 
  Let $m'$ be a finite $\delta$-word such that $m'\in M$. 
  Due to the Shrinking Lemma 
  \ref{lemma: Shrinking for ordinal automata}, 
  there is an $m\in M$ with $\supp(m)\subseteq \omega^{n+2} \cup
  \kappa$, i.e., a word whose support has a $\omega^\omega \eta$ gap
  at $\omega^{n+2}$. 
  Now let $\Ac_{\varphi}$ be the automaton that corresponds in
  $\mathfrak{M}$ to the $\Pi_2$ formula $\varphi(x)$ saying
  $x\in M \land $ $M$ \emph{forms a Boolean algebra without
  atoms}.\footnote{Note that associativity, commutativity, identity,
    distributivity are $\Pi_1$-statements, existence of complements is
    $\Pi_2$ and absence of atoms is a $\Pi_2$-statement.}
  Due to Lemma \ref{lemma: elementary substructure} and since 
  $m$
  satisfies $\varphi$ in $\mathfrak{M}$, 
  in the structure 
  $\mathfrak{M}'=(M', \cup',\cap,\bf{o}',\bf{1}')$ induced by
  $\bar\Ac$ seen as  $(\omega^{m+2} + \omega^{m+2} + \kappa)$-automata,
  there is a word $m'\in M'$  satisfying $\varphi$. 
  Thus, $\mathfrak{M}'$ forms a countable Boolean algebra
  without atoms. 
  By definition of $\kappa$,
  there is a  $k'\in \N$ such that $\kappa<\omega^{k'}$ 
  Set $k:=\max(k'+1,m+3)$. 
  Since there is an $\omega^{k}$-automaton marking position
  $\omega^{m+2} 2 + \kappa$ by a unique state, the countable atomless
  Boolean algebra 
  $\mathfrak{M}'$ also has an $\omega^k$-automatic presentation.
\qed \end{proof}

We finally show that the countable atomless Boolean algebra has no
$\omega^k$-automatic presentation.

\begin{theorem} The countable atomless Boolean algebra is not $\delta$-automatic for any ordinal $\delta$. 
\end{theorem} 

\begin{proof} 
  Assume that the countable
  atomless Boolean algebra has an $\omega^k$-automatic presentation
  $\mathfrak{M}=(M,\cup,\cap,\bf{0},\bf{1})$. Suppose the automata
  have at most $m$ states.  

  We follow the proof of \cite[Lemma
  3.4]{DBLP:journals/lmcs/KhoussainovNRS07}. We construct trees $T_n$
  with nodes $a_{\sigma}$ for all $\sigma\in 2^{\leq n}$ such that
  $T_n$ has exactly $2^n$ leaves and $u\cap v=\bf{0}$ for
  any any two leaves $u\neq v$ of $T_n$.\footnote{In this construction
    we replace Khoussainov et al.'s use of the length-lexicographic
    order by the use of an $\omega^k$-automatic well-order of the
    finite $\omega^k$-words.}
  The partial functions which
  determine the successor nodes $a_{\sigma0},a_{\sigma1}$ from
  $a_{\sigma}$ are definable in
  $\mathfrak{M}$ by first-order formulas with the quantifier
  $\exists^{\infty}$ and hence recognisable by automata by the closure
  of $\omega^k$-automatic relations under first-order definable
  relations\footnote{as usual in automatic structures 
    $\exists^{\infty}$ can be replaced by a
    first-order statement over some automatic expansion.} (see
  \cite{BedonBesCartonRispal2010}). Suppose each of these automata has
  at most $l\geq m$ states.  
  Then $\mathsf{supp}(a_{\sigma})\subseteq
  U_{l+1}^n(\mathsf{supp}(a_{\emptyset}))$ for all $\sigma \in
  \{0,1\}^{\leq n}$. 
  If $s=s_1\cup ...\cup s_j$ with pairwise
  different leaves
  $s_1,...,s_j\in T_n$, then $j\leq2^n$ and
  $\mathsf{supp}(s)\subseteq  U_{l+1}^{n 
    \lceil \log(2^n)\rceil}(\mathsf{supp}(a_{\emptyset}),\omega^l)$
  by the growth lemma for monoids. There are at most
  $|\Sigma_{\diamond}|^{(c_l(\mathsf{supp}(a_{\emptyset})\cup\{\omega^l\})+
    n\lceil\log(2^n)\rceil 
    l)^{l+1}}$ many such $s$. However, since the leaves of $T_n$ are
  pairwise incompatible, there are $2^{(2^n)}$ many $s$.
  This is a contradiction for large $n$.  
\qed \end{proof} 

Note that the growth argument in the previous proof can also be applied directly to the $\delta$-automatic presentation. 

\begin{lemma} Suppose $\Lf$ is of the form $\mathbb{Z}\cdot
\mathfrak{M}$ for some linear order $\mathfrak{M}$. Then the countable
atomless Boolean algebra is not $\Lf$-automatic.
\end{lemma}

\begin{proof} Suppose that $R$ is a locally finite $\Lf$-automatic
relation on an $\Lf$-automatic structure $M$ such that $R$ and the domain and relations of $M$ are recognised by automata
with $\leq m$ states. Then for any $(x,y)\in R$ and any $t\in
\mathsf{supp}(y)$, there are $u\leq v \in \mathsf{supp}(x)$ with
finite distance such that $\bar{u}\leq t\leq \bar{v}$ for the $n^{th}$
predecessor $\bar{u}$ of $u$ and the $m^{th}$ successor $\bar{v}$ of
$v$. The same growth rate argument as for finite automata (see \cite[Theorem 3.4]{DBLP:journals/lmcs/KhoussainovNRS07}) shows that
every infinite $\Lf$-automatic Boolean algebra is a finite product of
the Boolean algebra of finite and co-finite subsets of $\N$ with inclusion. 
\qed \end{proof}

\begin{question} Is the countable atomless Boolean algebra
  $\Lf$-$o$-automatic for any linear order $\Lf$ and any oracle $o$? 
\end{question}

\end{document}